\documentclass[twoside,11pt]{article}

\usepackage{jmlr2e}
\usepackage{amsmath}
\usepackage{algorithmic,algorithm}
\usepackage{caption,subcaption}

\jmlrheading{1}{2022}{1-48}{4/00}{10/00}{meila00a}{Jiahui, Chaoxia, Xinyu, Guohua and Alan}

\ShortHeadings{Model Averaging for Support Vector Machine by Cross-Validation}{Jiahui, Chaoxia, Xinyu, Guohua and Alan}
\firstpageno{1}

\allowdisplaybreaks

\newcommand{\calB}{{\cal B}}
\newcommand{\calA}{{\cal A}}
\renewcommand{\hat}{\widehat}

\def\CV{{\text{CV}}}

\def\text{\textrm}

\def\Exp{\mathbb{E}}

\newcommand{\bbeta}{\mbox{\boldmath${\beta}$}}

\newcommand{\bSigma}{\mbox{\boldmath${\Sigma}$}}

\newcommand{\bmu}{\mbox{\boldmath${\mu}$}}

\newcommand{\w}{{\bf w}}

\newcommand{\x}{{\bf x}}

\newcommand{\y}{{\bf y}}

\newcommand{\bu}{{\bf u}}
\newcommand{\sums}{\sum_{s=1}^{S_n}}
\newcommand{\sumin}{\sum_{i=1}^{n}}

\newcommand{\sumj}{\sum_{j=1}^{J}}

\newcommand{\tbeta}{\widetilde{\bbeta}}

\newcommand{\hbeta}{\hat{\bbeta}}
\newcommand{\sbeta}{\bbeta^\ast}
\newcommand{\spbeta}{\bbeta^{\ast +}}
\newcommand{\wbeta}{\widetilde{\bbeta}}

\newcommand{\J}{\mathbf{J}}
\newcommand{\Hess}{\mathbf{H}}
\newcommand{\onen}{\frac{1}{n}}

\newcommand{\Dn}{\mathcal{D}_n}
\def\tt{^{\rm T}}
\newcommand{\diff}{\mathrm{d}}

\def\eqb{\begin{eqnarray}}
\def\eqe{\end{eqnarray}}

\def\Pr{\mathrm{Pr}}

\def\calW{{\mathcal{W}}}

\newcommand{\supw}{\sup_{\w\in\calW}}
\newcommand{\supjw}{\sup_{\w\in\calW^{(j)}}}
\newcommand{\wj}{\mathbf{w}^{(l)}}
\newcommand{\infw}{\inf_{\w\in\calW}}

\def\s{{(s)}}

\def\boxit#1{\vbox{\hrule\hbox{\vrule\kern6pt\vbox{\kern6pt#1\kern6pt}\kern6pt\vrule}\hrule}}

\newcommand{\sups}{\max_{1\leq s\leq S_n}}
\newcommand{\infs}{\min_{1\leq s\leq S_n}}
\newcommand{\supi}{\max_{1\leq i \leq n}}
\newcommand{\supjN}{\max_{1\leq l\leq N}}

\newcommand{\pmax}{{p_{\max}}}
\newcommand{\sumsn}{\sum_{s=1}^{S_n}}
\newtheorem{condition}{Condition}

\newtheorem{lem}{\bf Lemma}

\newcounter{subassumption}[asu]

\graphicspath{{figures/}}
\begin{document}

\title{Model Averaging for Support Vector Machine by Cross-Validation}

\author{\name Jiahui Zou \email zoujiahui@cueb.edu.cn \\
       \addr School of Statistics\\
       Capital University of Economics and Business\\
       Beijing 100070, China
       \AND
       \name Chaoxia Yuan \email chaoxiayuan@163.com \\
       \addr School of Statistics\\
       East China Normal University\\
       Shanghai 200062, China
        \AND
       \name Xinyu Zhang \email xinyu@amss.ac.cn \\
       \addr Academy of Mathematics and Systems Science\\
        Chinese Academy of Sciences\\
         Beijing 100190, China
       \AND
       \name Guohua Zou \email ghzou@amss.ac.cn \\
       \addr School of Mathematical Sciences\\
        Capital Normal University\\
         Beijing 100048, China
       \AND
       \name Alan T. K. Wan \email Alan.Wan@cityu.edu.hk \\
       \addr  Department of Management Sciences and School of Data Science\\
        City University of Hong Kong\\
        Kowloon, Hong Kong}

\editor{}

\maketitle

\begin{abstract}
Support vector machine (SVM) is a well-known statistical technique for classification problems in machine learning and other fields.
An important question for SVM is the selection of covariates (or features) for the model.
Many studies have considered model selection methods. As is well-known, selecting one winning model over others can entail considerable instability in predictive performance due to model selection uncertainties.  This paper advocates model averaging as an alternative approach, where estimates obtained from different models are combined in a weighted average.  We propose a model weighting scheme and provide the theoretical underpinning for the proposed method.  In particular, we prove that our proposed method yields a model average estimator that achieves the smallest hinge risk among all feasible combinations asymptotically.  To remedy the computational burden due to a large number of feasible models, we propose a screening step to eliminate the uninformative features before combining the models.   Results from real data applications and a simulation study show that the proposed method generally yields more accurate estimates than existing methods.
\end{abstract}


\section{Introduction}\label{sec:intr}

Support vector machine (SVM) \citep{1995Vapnik,Sch2003Learning} is a well-known statistical technique for handling classification problems in biology, machine learning, medicine and many other fields.  An important question for SVM is how to select the covariates (or features) of the model.  Many studies have attempted to address this question and come up with a variety of methods.   \cite{Jason2000Feature} proposed a scaling method. \cite{Guyon2002Gene} suggested a recursive feature elimination procedure.  Some authors have considered regularisation methods within the context of SVM.  For example, \cite{Bradley1998Feature}, \cite{Ji20041} and \cite{Marten2011Support} investigated the properties of a $L_1$ penalised SVM; \cite{Li2006The} considered SVMs with $L_1$ and $L_2$ penalties; \cite{Zou2008The} considered the $L_\infty$ penalised SVM in the presence of prior knowledge in  the group information of features;
\cite{Hao2006Gene} and \cite{Becker2011Elastic} suggested a non-convex penalty in the application of gene selection; and
\cite{Park2012Oracle} and \cite{SVMselect2014}
investigated the oracle property of the SCAD-penalised SVM with a fixed number of covariates under a class of non-convex penalties. As well, \cite{Zhang2016aconsistent} developed a consistent information criterion for SVM with divergent-dimensional covariates. \cite{Claeskens2008An} proposed an information criterion called $\text{SVMIC}_L$ for feature selection in SVM, and \cite{Zhang2016aconsistent} proved that this criterion results in model selection consistency when the number of features is finite, and developed a modified version of the criterion that achieves model selection consistency when the number of features diverges at an exponential rate of the sample size.  The above provides an overview of the current state of the literature even though the list is by no means exhaustive.

A major critique of model selection, is that as this practice focuses on one single model and ignores all others, it can lead to decisions riskier than warranted.  Model averaging is an alternative approach that has been proposed to address the above issue.  Unlike model selection that commits to one champion model and discounts all others in the pool, model averaging combines different candidate models by an appropriate weighting scheme. A major takeaway from existing studies, is that model averaging often results in improved predictive accuracy and is a more robust strategy than model selection \citep{hansen2007least,wan2010least}.

Of the model averaging methods, Bayesian model averaging (BMA) is a common choice of technique as it is straightforward to implement. See \cite{hoeting.madigan.ea:1999} for a review of the BMA methodology.  The major challenge confronting BMA is choosing subjective priors.  Frequentist model averaging (FMA), which is a more recent vintage, avoids the difficulty of specifying prior probabilities as it is entirely data-driven.  Many weight choice methods within the FMA paradigm, including
the smoothed information criteria \citep{Buckland1997Model,claeskens.croux.ea:2006}, optimal weighting \citep{hansen2007least,Zhang2016Optimal,zhang2019Parsimonious}, and adaptive weighting \citep{Yuan2005Combining,ZhangAdaptively} have been proposed in a variety of contexts.  To the best of our knowledge, no study has considered model averaging within the context of SVM, and the purpose of this paper is to take steps in this direction.


Our contribution is two-fold.   First, we develop a weight choice criterion and prove that the resultant model average estimator is asymptotically optimal
in the sense of achieving the smallest hinge risk among all feasible combinations asymptotically.  It is worthwhile to mention that our analysis allows the hinge loss to be non-smooth as well as asymmetric.
Second, to remedy the computational burden due to a large number of feasible models, we propose a screening step to eliminate the uninformative features before combining the models.

The rest of the paper is organised as follows. Section \ref{sec:modelsetup} describes the model setup and introduces the FMA method. Section \ref{sec:theory} presents results on the theoretical properties of the resultant FMA estimator. In Section \ref{sec:simulation}, we evaluate the usefulness of the proposed procedure in finite samples. Section \ref{sec:data} applies the method to three real data sets.  Proofs of results are relegated to the Appendix.

\section{Model setup and estimation method}\label{sec:modelsetup}
\subsection{SVM and model averaging}
Consider a random sample $\Dn=\{(\x_1, y_1),(\x_2, y_2),...,(\x_n, y_n)\}$, where $y_i\in\{-1,+1\}$ and each of $(\x_i,y_i)$ is independently drawn from an identical distribution.
Denote $\x_i=(1,x_{i1},x_{i2},...,x_{ip})\tt\in\mathcal{R}^{p+1}$, and $\x_i^+=(x_{i1},x_{i2},...,x_{ip})\tt$.  Write $\bbeta=(\beta_0,\beta_1,...,\beta_p)\tt\in \mathcal{R}^{p+1}$ and
$\bbeta^+=(\beta_1,\beta_2,...,\beta_p)\tt\in\mathcal{R}^{(p+1)}$, where $\bbeta^+$ is the coefficient vector corresponding to $\x_i^+$.
The objective of linear SVM is to find a hyperplane, defined by $\x\tt\bbeta=0$, to draw a boundary between $y=1$ and $y=-1$.  This hyperplane is commonly estimated by solving the following optimisation problem \citep{Trevor2005The}:
\begin{align}\label{eq:model}
\min_{\scriptsize\bbeta}\left\{\onen\sumin \left(1-y_i\x_i\tt\bbeta\right)_++\frac{\lambda_n}{2}\left\|\bbeta^+ \right\|^2\right\},
\end{align}
where
 $\|a\|$ is the Euclidean norm operator of the vector $a$, $(1-t)_+=\max(1-t,0)$ is a hinge loss function that can be asymmetric or symmetric, and smooth or non-smooth, and $\lambda_n$ is a tuning parameter.

One usually tackles the problem concerning the uncertainty in $\x$ by model selection, as discussed in Section \ref{sec:intr}.  Here, we consider the alternative strategy of model averaging that combines models with different covariates.
It is assumed that each model contains a minimum of one covariate in addition to the intercept term.  Hence there exist a maximum of $2^{p}-1$ feasible candidate models.   Some of these models may be uninformative and one may consider removing them before averaging.  Without loss of generality, assume there are $S_n$ models to be combined.  Clearly, it is required that $S_n \leq 2^{p}-1$.
Let $s\in\{1,...,S_n\}$, and denote $\mathcal{M}_s=\{j_1,...,j_{p_s}\}\subset\{1,...,p\}$ as the set consisting of the indices of elements of $\x_{\s,i}$.
For the $sth$ model,  $\x_{\s,i}=(1,x_{\s i,j_1},...,x_{\s i,j_{p_s}})\tt$ and $\bbeta_\s\tt=(\beta_{\s 0},\beta_{\s j_1},...,\beta_{\s j_{p_s}})$.  The estimator of $\bbeta_\s$ is obtained by solving the optimisation problem described in (\ref{eq:model}), replacing $\bbeta$ by $\bbeta_\s$ and $\x_i$ by $\x_{\s i}$ everywhere, yielding
\begin{align}\label{eq:candidatemodel}
\hbeta_\s=\arg\min_{\scriptsize\bbeta_\s}\left\{\onen\sumin \left(1-y_i\x_{\s i}\tt\bbeta_\s\right)_++\frac{\lambda_n}{2}\left\|\bbeta^+_\s \right\|^2\right\}.
\end{align}

Following  \cite{Koo2008Abahadur}, we denote $\bbeta^*_\s$ as the ``quasi-true'' parameter\footnote{The parameter that minimises the population hinge loss is the ``quasi-true" parameter when the working model is not identical to the true data generating process. If the two are identical, the ``quasi-true" parameter is the true parameter. } that minimises the population hinge loss. That is,
\begin{align}
\bbeta^*_\s=\arg\min_{\scriptsize \bbeta_\s}\Exp(1-y\x_\s\tt\bbeta_\s)_+.
\end{align}
To facilitate analysis, let $\Pi_s$ be a $(p+1)\times (p_s+1)$ dimensional selection matrix consisting of 1 or 0 and write $\hbeta_s=\Pi_s\hbeta_\s$.  For example, if the covariate vector of $sth$ model is $\x_ {\s,i}=(1,x_{\s,i3})\tt$, and
 $\hbeta_\s=(1,2)\tt$, then $\Pi_s=\begin{pmatrix}
                             1 & 0 & 0 & 0 \\
                             0 & 0 & 0 & 1
                           \end{pmatrix}\tt$
and $\hbeta_s=(1,0,0,2)\tt$.
The model average estimator of $\bbeta$ is a weighted sum of  $\hbeta_s$'s, $s=1,\cdots,S_n$, i.e.,
\begin{align}\label{eq:4}
\hbeta(\w)=\sums w_s\hbeta_s,
\end{align}
where $\w=(w_1,...,w_{S_n})\tt$ is the weight vector belonging to the set $\calW=\{\w\in[0,1]^{S_n}: \sums w_s=1 \}$. We label $\hbeta(\w)$ as the SVM model average (SVMMA) estimator.

\subsection{Weight choice criterion}
As discussed above, the standard SVM approach derives the coefficient estimates by minimising the hinge loss associated with a given model.  Analogously, when more than one model is involved, the hinge loss may be modified to be
\begin{align} \label{fmahinge}
\onen\sumin \left(1-y_i\x_i\tt\hbeta(\w) \right)_+.
\end{align}
The purpose is to find a weight vector $\w$ to be used in (\ref{eq:4}) such that the resultant SVMMA estimator yields an optimal property.  Clearly, the hinge loss in (\ref{fmahinge}) favours bigger models, and if one minimises (\ref{fmahinge}) directly, over-fitting becomes a distinct possibility. To reconcile this issue, we focus on the following alternative out-of-sample risk as an alternative:
\begin{align} \label{outhinge}
R_n(\w)&=\Exp\left\{(1-\tilde{y}\tilde{\x}\tt\hbeta(\w))_+ |\Dn\right\},
\end{align}
where $(\tilde{y},\tilde{\x})$ is an independent copy from the distribution of $\Dn$.  It is instructive to note that the expectation in (\ref{outhinge}) is computed with respect to $(\tilde{y},\tilde{\x})$, but $\hbeta(\w)$ is estimated based on $\Dn$.
As $\Dn$ is unknown, we minimise the following estimator of (\ref{outhinge}) to obtain $\w$:
\begin{align}\label{eq:7}
Q_m(\w)=\frac{1}{m}\sum_{i=1}^m \left(1-\tilde{y}_i\tilde{\x}_i\tt\hbeta(\w) \right)_+.
\end{align}
where $\{(\tilde{y}_i,\tilde{\x}_i)\}_{i=1}^m$ are independent copies from the same distribution of $\Dn$.  In addition, we divide the data into a training sample
and a validation or test sample.
This is the cross-validation (CV) approach that allows the estimated model to be tested on new data. Let $n$ and $n_{test}$ be the size of the training sample and test sample respectively and
 $J$ the number of folds associated with the CV approach such that the number of observations in each block is $[n/J]$, where $[g]$ is the truncated integer value of $g$. Denote $\calA(j)=\{(j-1)M_n+1,(j-1)M_n+2,...,jM_n\}$, $|\calA(j)|$ the cardinality of $\calA(j)$, and $\calB(j)=\{1,2,...,n\}\bigcap \calA(j)^c$.   The CV approach is based on the criterion
\begin{align}
\CV_n(\w)&=\frac{1}{n}\sum_{j=1}^{J} \sum_{i\in \calB(j) }\left(1-y_i\x_i\tt \wbeta^{[-j]}(\w)\right)_+,\label{eq:CV}
\end{align}
where
\begin{align*}
\wbeta^{[-j]}(\w)
=\sums w_s \Pi_s\wbeta_\s^{[-j]}
\text{\quad and\quad }
\wbeta_\s^{[-j]}=\arg\min_{\scalebox{0.7}{$\bbeta_\s$}} \frac{1}{|\calA(j)|}\sum_{i\in \calA(j)}(1-y_i\x_{\s,i}\tt\bbeta_\s)_+.
\end{align*}
Some explanations of the CV criterion and the above notations are in order.  We denote $\wbeta_\s^{[-j]}$ as the estimator obtained with the $j$th sub-sample of data removed from the training sample; $\wbeta_\s^{[-j]}$ is obtained by minimising the hinge loss averaged over $|\calB(j)|$ observations. The integrated estimator $\wbeta^{[-j]}(\w)$ is obtained by combining $\wbeta_\s^{[-j]}$ obtained from each of the $S_n$ models.  To evaluate the performance of $\wbeta^{[-j]}(\w)$ on the test data, we calculate the associated hinge loss on the $j$th sub-sample of data that contains $|\calB(j)|$ observations. We repeat this process for all $J$ sub-samples and the optimal weight vector is obtained by a minimisation of (\ref{eq:CV}).

Although model averaging can often deliver more precise estimates and reduce bias compared to model selection, it is computationally intensive, especially  when the data dimension is high. With $p$ covariates, there are $2^p-1$ potential candidate models, and when $p$ is large, it is difficult if not impossible to combine all models. We mitigate this problem by screening out the uninformative covariates before combining the models.  Our model screening procedure entails sorting the covariates under $L_1$ penalty. We order the covariates according to the sequence in which the estimated coefficient of the covariate becomes non-zero as the penalty parameter decreases.  Based on this ordering, we construct $S_n$ candidate models by including
one extra covariate successively for each new model such that a given model is always nested within the next smallest model.   The details are described in Algorithm \ref{alg:weights}.
 \cite{zhang2019Parsimonious} considered a similar model screening method.  After constructing the $S_n$ models, we calculate the model weights by the CV criterion in \eqref{eq:CV}, and combine the models in accordance with \eqref{eq:4}.  For the choice of $J$, we find that it generally has little impact on the results and we suggest choosing $J=5, 10$.  Finally, we follow the steps of Algorithm \ref{alg:prediction} to predict $y$.
\\

\begin{algorithm}
\caption{Calculate $\hat{\w}$}
\label{alg:weights}
\begin{algorithmic}
\REQUIRE $D=\{(\x_1,y_1),(\x_2,y_2),...,(\x_n,y_n)\}$\hspace{\fill} \#The first element of $\x_i$ is 1.\\
\hspace{-0.42cm}\textbf{Step 1:} Pre-screening
\REQUIRE $a$, $b$, $L$, $\mathcal{C}_0=\emptyset$, $I=$list(), $J$, $S_n$, $M_n=[n/J]$\\
 \hspace{\fill}{
\# The $L_1$-penalty parameter $\lambda$ is in $[a,b]$ and we let $a$ be 0.001. \\
 \hspace{\fill} \# $J$ is the number of folds in CV and $S_n$ is the number of candidate models.}
\FOR {$l=0$ to  $L$}
\STATE $\lambda=a+l(b-a)/L$
\STATE $\breve{\bbeta}=\arg\min_{\bbeta} n^{-1}\sum_{i=1}^{n}(1-y_i\x_i\tt\bbeta)_++\lambda\|\bbeta^+\|_1$
\STATE $\mathcal{C}_{l+1}=\{i|\breve{\beta}_i^+=0,i=1,2,...,p\}$\hspace{\fill}\# $\breve{\bbeta}=(\breve{\beta}_0,\breve{\bbeta}^{+\mathrm{T}})\tt, \breve{\bbeta}^+=(\breve{\beta}_1^+, \breve{\beta}^+_2,..., \breve{\beta}_p^+)\tt$.
\STATE $I=I.extend(list(\mathcal{C}_{l+1}\backslash\mathcal{C}_{l}))$\hspace{\fill} \# { I.extend(.) means adding all the elements of a list to the tail of I}
\ENDFOR\\
\STATE $I=I.extend(list(0))$\hspace{\fill}{ \# Add the index of 1 in $\x_i$}
\STATE $I=I.reverse()$\hspace{\fill} \# { I.reverse(.) means reversing the elements of I}\\
\hspace{-0.42cm}\textbf{Step 2:} Estimate each candidate model
\FOR{$j=1$ to $J$}
\STATE
 $\calA(j)=\{(j-1)M_n+1,(j-1)M_n+2,...,jM_n\}$\\ $\calB(j)=\{1,2,...,n\}\bigcap \calA(j)^c$\\
    \FOR {$s=1$ to $S_n$}
    \STATE $\hbeta_\s=\arg\min_{\scriptsize\bbeta_\s}\left\{(n-M_n)^{-1}
    \sum_{i \in \calA(j)} (1-y_i
    {\x_{i,I[0:s]}}\tt\bbeta_\s)_++\{2(n-M_n)\}^{-1}\|\bbeta^+_\s \|^2\right\}$
    \hspace{\fill} \# $\x_{i,I[0:s]}$ means collecting the elements of $\x_i$ with indices in $I[0:s]$
    \STATE $\hbeta^{(j)}_s=\Pi_s\hbeta_\s$\hspace{\fill}\#$\Pi_s$ is the selection matrix
    \ENDFOR\\
\ENDFOR\\
\hspace{-0.42cm}\textbf{Step 3:} Calculate $\hat{\w}$
\STATE $\hat{\w}=\arg\min_{\w}{n}^{-1}\sum_{j=1}^{J} \sum_{i\in \calB(j) }\left(1-y_i\x_i\tt \sums w_s\hat{\bbeta}^{(j)}_s\right)_+$
\RETURN $\hat{\w}$
\end{algorithmic}
\end{algorithm}

\begin{algorithm}
\caption{Prediction}
\label{alg:prediction}
\begin{algorithmic}
\REQUIRE $\hat{\w}$, $\x_{\text{new}}$,$I$, $S_n$, $D=\{(\x_1,y_1),(\x_2,y_2),...,(\x_n,y_n)\}$\\
\STATE Estimate each candidate model based on $I$, $S_n$ and $\Dn$ to obtain $\{\hbeta_1,\hbeta_2,...,\hbeta_{S_n}\}$\\
\RETURN $\text{sign}(\x_{\text{new}}\tt\sums \hat{w}_s \hat{\bbeta}_s)$
\end{algorithmic}
\end{algorithm}

\section{Theoretical justification}\label{sec:theory}
\subsection{Notations and technical conditions}
This section is devoted to an investigation of the theoretical properties of the proposed model averaging strategy. Denote $L_\s(\bbeta_\s)=\Exp(1-y\x_\s\tt\bbeta_\s)_+$, $\J_s(\bbeta_\s)=-\Exp\left(\mathbf{1}_{\{1-y\x_\s\tt \scalebox{0.5}{\bbeta}_\s\geq 0\}}y\x_\s \right)$ and $\Hess_s(\bbeta_\s)=\Exp\{\delta(1-y\x_\s\tt
\bbeta_\s)\x_\s\x_\s\tt \}$, where $\mathbf{1}_{\{\cdot\}}$ is the indicator function and $\delta(\cdot)$ is the Dirac delta function, $s=1,2,...,S_n$.  \cite{Koo2008Abahadur} showed that under some regular conditions, $\J_\s(\bbeta_\s)$ and $\Hess_\s(\bbeta_\s)$ possess the mathematical properties of  the gradient and Hessian matrix of $L_\s(\bbeta_\s)$ respectively. In addition, we let $f_+$ and $f_-$ be the densities of $\x^+\in\mathcal{R}^p$ conditional on $y=1$ and $y=-1$, respectively.  Note that the dimension of the covariates $\x_{\s,i}$ used in the $sth$ candidate model is $p_s+1$. We therefore write $\pmax=\max_{1\leq s\leq S_n} p_s+1$ as the dimension of the largest candidate model.

Our proofs of theoretical results require the following conditions:

\begin{condition}
 \label{con:1i}$f_+$ and $f_-$ are continuous and have the same common support in $\mathcal{R}^p$.
 \end{condition}
 \begin{condition}
 \label{con:1ii}
      There is a constant $C_1>0$ such as $\supi\max_{1\leq j\leq p} |x_{ij}|<C_1$.
\end{condition}
\begin{condition}
 \label{con:1ii2} For $s=1,2,...,S_n$, the $sth$ candidate model has a unique ``quasi-true" parameter $\sbeta_\s$ and there exists a constant $C_2>0$ such that $\|\sbeta_
     \s\|\leq C_2 \sqrt{p_s}$.
\end{condition}
\begin{condition}
  \label{con:1iv} The densities of $\x_{(s), i}\tt\sbeta_{(s)}$ conditional on $y=1$ and $y=-1$ are uniformly bounded away from zero, and have a uniform upper bound $C_3$ (a positive constant) at the neighborhood of $\x_{(s),i}\tt\sbeta_{(s)}=1$ and $\x_{(s),i}\tt\sbeta_{(s)}=-1$ respectively.
\end{condition}
\begin{condition}
\label{con:1viii} For $s=1,2,...,S_n$, there exists a positive constants $c_0$ 
       such that $\infs\\\lambda_{\min}\{\Hess_s(\sbeta_\s)\}\geq c_0$,
      where $\lambda_{\min}(\cdot)$ is the smallest eigenvalue of a matrix $(\cdot)$.
\end{condition}
\begin{condition}
\label{con:1v} $\pmax=O(n^\kappa)$ for some constant $\kappa\in (0,1/5)$.
\end{condition}

Condition \ref{con:1i}, which is adopted from \cite{Koo2008Abahadur}, ensures that $\J_s(\bbeta_\s)$ and $\Hess_s(\bbeta_\s)$ are well-defined.   Condition \ref{con:1ii} facilitates the measurement of the order of $\x_{\s,i}$.  This is a common condition in high-dimensional studies \citep[e.g.,][]{Wang2012Quantile, Lee2014Model}. Condition \ref{con:1ii2} is a mild condition that guarantees the existence of the ``quasi-true" parameter. Similar conditions can be found in \cite{White1982}, \cite{Zhang2016Optimal} and \cite{Ando2017}. Condition \ref{con:1iv} assumes that as the sample size increases, there is information around the non-differentiable points of the hinge loss function to enable the boundary of hyperplane to be identified - note that the observations that satisfy $\x_{(s),i}\tt\sbeta_{(s)}=1$ or $\x_{(s),i}\tt\sbeta_{(s)}=-1$ are around the hyperplane's boundary, and are usually the non-differentiable points of the hinge loss function. We require the densities to be bounded away from zero so that there is information available to identify the hyperplane boundary. On the other hand, the data points should avoid being too concentrated near the boundary or non-differentiable points and accordingly we control the densities by the constant $C_3$.   This is similar to the condition  for model selection consistency of non-convex penalised SVM in high-dimension \citep{SVMselect2014}.  Condition \ref{con:1viii} assumes that the Hessian matrix is well-behaved and nonsingular when $\bbeta_\s$ is near the ``quasi-true" parameter.
Condition \ref{con:1v} allows the dimension of the covariates to diverge with the sample size and imposes a restriction on its rate of divergence.  Specifically, as
 the convergence of $\hbeta_\s$ is only related to the number of covariates of the $sth$ model, we impose restrictions on $\pmax$ and not on $p$. As $\pmax<p$,  Algorithm \ref{alg:weights} can handle the cases of $p>n$ and $p/n\to \infty$.

\subsection{Theoretical results}

\begin{lem}\label{lem:consist}
 Under Conditions \ref{con:1i}-\ref{con:1v}, if $S_n=O\{\exp(n^{\tau})\}$ for some constant $\tau\in (0, 1/2-3\kappa/2)$ and { $\lambda_n=O(\sqrt{n^{-1}\log(\pmax)})$}, we have
 \begin{align}
 \sups \left\|\hbeta_\s-\sbeta_\s \right\|={ O_p\left(\sqrt{\frac{\pmax\log(\pmax)}{n}}\right)},\label{eq:8}
 \end{align}
and
 \begin{align}
\max_{1\leq j\leq J} \sups \left\|\tbeta_\s^{[-j]}-\sbeta_\s \right\|={ O_p\left(\sqrt{\frac{\pmax\log(\pmax)}{n}}\right)}.\label{eq:9}
 \end{align}
\end{lem}
This lemma provides the speed in which the estimates $\hat{\bbeta}_\s$ and $\tbeta_\s^{[-j]}$ converge to  $\bbeta^*_\s$ uniformly for $s$.   As the covariates are contained in different candidate models and we estimate the parameters of each model independently, it suffices to explore the relationship between $n$ and $\pmax$ instead of  $n$ and $p$.

\begin{condition}\label{con:xi}
There exists a constant $\xi_0$ such that
\begin{align}
\liminf_{n\to\infty}\xi_n\geq\xi_0>0,
\end{align}
where $\xi_n=\infw R_n(\w)$.
\end{condition}
 This condition is readily satisfied because the hinge loss is nonnegative and typically the data cannot be distinctly separated by the linear hyperplane. Similar conditions are often used in other studies of model averaging, such as Condition (A.6) of \cite{hansen.racine:2012} and Condition (A3) of \cite{Ando2017}.
\begin{theorem}\label{thm:optimality}
Under Conditions \ref{con:1i} - \ref{con:xi}, if $S_n=O(n^\tau)$ for some constant $\tau\in (0,1-2\kappa)$, then
\begin{align}
\frac{R_n(\hat{\w})}{\inf_{\w\in\mathcal{W}} R_n(\w)}\to 1\label{eq:opt}
\end{align}
in probability, where $\hat{\w}$ is the optimal solution in \eqref{eq:CV}.
\end{theorem}

This theorem shows that the SVMMA estimator is asymptotically optimal in the sense that it results in a hinge risk that is asymptotically identical to that obtained from the infeasible best possible model average estimator. In contrast to Lemma \ref{lem:consist}, we allow the order of $S_n$ to be $O(n^\tau)$ instead of $O\{\exp(n^\tau)\}$.


\section{A simulation study}\label{sec:simulation}
\subsection{Methods for comparison and evaluation criteria}
The purpose of this section is to examine the performance of the SVMMA estimator under sample sizes commonly encountered in practice via a simulation study. We include the following competing methods in the comparison:
\begin{itemize}
\item The SVM information criterion (SVMICL) and its modified high-dimensional version (SVMICH) introduced by \cite{Zhang2016aconsistent}, defined as
\begin{align}\label{eq:SVMICl}
\text{SVMICL}_s=\sumin (1-y_i\x_{\s,i}\tt\hbeta_\s)_++p_s\log(n)
\end{align}
and
\begin{align}\label{eq:SVMICh}
\text{SVMICH}_s=\sumin(1-y_i\x_{\s,i}\tt\hbeta_\s)_++\log^{3/2}(n) p_s
\end{align}
respectively.  The SVMICL and SVMICH select the model with the smallest value of their respective criterion.

\item The smoothed-SVMICL (SCL) and smoothed-SVMICH (SCH) methods which are model averaging counterparts to the SVM and SVMICL respectively.  The SCL and SCH weights for the $sth$ model are given by
\begin{align}
\text{SCL}_s=\exp\left(-\text{SVMICL}_s/n\right)\Big/\sumsn\exp\left(-\text{SVMICL}_s/n\right), s=1,2,...,S_n,
\end{align}
and
\begin{align}
\text{SCH}_s=\exp\left(-\text{SVMICH}_s/n\right)\Big/\sumsn\exp\left(-\text{SVMICH}_s/n\right), s=1,2,...,S_n,
\end{align}
respectively.
\item The bagging \citep{Breiman1996Bagging} (BAG) and adaboosting \citep{Freund1997A} (ADA) methods that belong to the class of ensemble learning methods, both being popular methods in machine learning research.   In the jargon of bagging and adaboosting, the candidate models are known as base learners. Bagging  combines the outputs from base learners. Adaboosting is a type of boosting developed for classification problems. Unlike model averaging, adaboosting places no constraint on the weights for the outcomes from base learners. While bagging focuses on reducing the variance, adaboosting emphasises bias reduction.
\item The uniform weighting method (UNIF) that assigns all the candidate models with an equal weight $1/S_n$.
\end{itemize}

We evaluate the performance of these methods by the following normalised hinge loss (NHL) and error rate on prediction (ER):
\begin{align}
\text{NHL}&=\frac{1}{D}\sum_{d=1}^{D}\frac{n_{\text{test}}^{-1}\sum_{i=1}^{n_{\text{test}}}\left(1-\tilde{y}^{(d)}_i \tilde{\x}^{(d)\mathrm{T}}_i\hbeta^{(d)}(\hat{\w}^{(d)})\right)_+}{ \min_{\w\in\calW}n_{\text{test}}^{-1}\sum_{i=1}^{n_{\text{test}}}\left(1-\tilde{y}^{(d)}_i\tilde{\x}^{(d) \mathrm{T}}_i\hbeta^{(d)}(\w)\right)_+},\label{eq:NHL}
\end{align}
and
\begin{align}
\text{ER}&=\frac{1}{D n_{\text{test}}}\sum_{d=1}^{D}\sum_{i=1}^{n_{\text{test}}} 1(\tilde{y}^{(d)}_i\neq\hat{y}^{(d)}_i),
\end{align}
where at the $dth$ replication, $\tilde{y}_i^{(d)}$ and $\tilde{\x}_i^{(d)}$ are observations of $y_i$ and $\x_i$ obtained from the test sample respectively,  $\hbeta^{(d)}(\hat{\w}^{(d)})$, $\hat{\w}^{(d)}$, and $\hat{y}^{(d)}_i$ are the estimate of $\beta$, the weight vector calculated by a given method,  and a predicted value respectively, $D$ is the number of replications, and $n_{\text{test}}$ is the size of test sample. Note that for a model selection method, the elements of $\hat{\w}^{(d)}$ are either 1 or 0.    As bagging and adaboosting only deliver the outcome and not the estimate of $\bbeta$, we omit them in NHL comparisons.

\subsection{Simulation designs}
We consider the following two data generating processes (DGPs) similar to those used in the simulation study of \cite{Zhang2016aconsistent}.  They are related to linear discriminant analysis and Probit model respectively.

\textbf{DGP1:} $\Pr(Y=1)=\Pr(Y=-1)=0.5, \x|(Y=1)\sim N(\bmu,\bSigma), \x|(Y=-1)\sim N(-\bmu, \bSigma),
 \bmu=(\underbrace{0.6,\cdots,0.6}_q,0,...,0)\tt\in \mathcal{R}^p$, $ \bSigma=(\sigma_{ij})_{p\times p}$, where $\sigma_{ii}=1$ for $i=1,2,...,p$, $\sigma_{ij}=0.2$ for $1\leq i\neq j\leq p$,  and $q$ is a tuning parameter that measures the sparsity.

\textbf{DGP2:} $\Pr(Y=1)=\Phi(\x\tt\bbeta)$, $\x\sim N({\bf{0}}_p,\bSigma)$,  $\bbeta=(\underbrace{2,\cdots,2}_q,0,...,0)\tt$, $\bSigma=(\sigma_{ij})_{p\times p}$ with $\sigma_{ii}=1$ for $i=1,2,...,p$ and $\sigma_{ij}=0.4^{|i-j|}$ for $1\leq i\neq j\leq p$, where $\Phi(\cdot)$ is the cumulative distribution function (CDF) of the standard normal distribution.

For the evaluation of the method's robustness with respect to model misspecification due to missing covariates,  we consider the following two scenarios:
\begin{itemize}
  \item \textbf{S1}: { All candidate models are misspecified;}
  \item  \textbf{S2}: { There exists at least one candidate model correctly representing the underlying DGP.}
\end{itemize}
We set the dimension $p$ in DGP1 and DGP2 to 1000.  For Scenario S1, we set $q=5$ and omit one distinct covariate from each model at the training stage; for S2,  we set $q=4$ and omit no variable. We apply the pre-screening step  in Algorithm \ref{alg:weights}, set $a=0.001$ and $b=10$ for the range of $\lambda$, and $L=50$ for ordering the covariates.  We choose $S_n$ such that the SVMMA results in the smallest risk as seen from the learning curves. For example, in Figure \ref{fig:Scenario11_learn_curves}, when $S_n$ is in $(75,160)$, the error rates of the SVMMA method for both the training and test samples are at their lowest. We set $S_n=100$ but other values within the range of $(75,160)$ can also be chosen. We let the candidate models constructed for the SVMMA method be the base learners for bagging and adaboosting. This is to facilitate a fair comparison between model averaging and the two ensemble learning methods.  We use a five-fold CV to calculate the model weights\footnote{We find that the number of folds generally has little effect on the performance of the method}, and set $D=200$ and the sizes of the training and test samples to $n=100,200,300,400$ and $n_{test}=10000$ respectively.

\subsection{Simulation results}
The results are shown in  Figures \ref{fig:Scenario11_learn_curves}-\ref{fig:correct-ER}, where $(.)$-Train and $(.)$-Test in Figures \ref{fig:Scenario11_learn_curves}-\ref{fig:Scenario12_learn_curves} represent the error rates of a given method in the training and test samples respectively.    Figures \ref{fig:Scenario11_learn_curves}-\ref{fig:Scenario12_learn_curves} show that generally speaking, the performance of SVMMA, bagging and adaboosting methods all improve as the number of base learners or candidate models increases, although in the case of adaboosting the improvement is not remarkable. For a given number of candidate models or base learners, SVMMA produces the best results in the majority of cases; while bagging is able to produce similar results to SVMMA, it can do so only at the expense of a larger number of base learners.  This is a notable advantage of the proposed model averaging approach over bagging.   Judging from the learning curves in \ref{fig:Scenario11_learn_curves}-\ref{fig:Scenario12_learn_curves}, the superior SVMMA results are achieved by setting $S_n$ in (75, 160).

Figures \ref{fig:miss-NHL}-\ref{fig:miss-ER} show that in terms of NHL, in all parts of the parameter space, the SVMMA approach delivers the best estimates, often by a large margin, the SCL estimator always has the edge over the SCH estimator, and both the SCL and SCH estimators dominate their corresponding model selection counterparts, the SVMICL and SVMICH estimators. Although the UNIF method is a distant second best compared to the SVMMA estimator, it is superior to the SCL estimator in terms of NHL over a large part of the parameter space. Exceptions occur when $n$ is very large, where the SCL can sometimes be slightly more accurate.

Generally speaking, the above findings carry over to comparisons in terms of ER.
In all cases, adaboosting performs poorly, at a level comparable to the SVMICH estimator.  On the other hand, bagging can sometimes deliver marginally better estimates than the SVMMA estimator when $n$ is small.  That being said, bagging is an inferior strategy compared to most other methods  when $n$ is moderate to large, while SVMMA offers more stable performance and is the preferred estimator in the majority of cases.   The rather erratic performance of bagging can be explained by noting that a good bagging method relies on many independent learners \citep{Zhou2012Ensembling}. However, in our case, every two candidate models have at least one common covariate, and as the training sample size grows, the diversity of their outcomes reduces. On the other hand, in the case of the SVMMA estimator, the weights reflect the strength of the candidate models, and better candidate models will be given higher weights, resulting in improved accuracy.  It is instructive to note that a model average estimator with a weight constraint may be interpreted  as a shrinkage estimator that balances between bias and variance \citep{Jagannathan2003Risk}.
%

\section{Real data examples} \label{sec:data}
In this section, we apply the proposed methods to three  real data sets, all obtained from the University of California at Irvine's Machine Learning Repository. In all cases, we standardise the covariates and randomly split the data into a training subset and a test subset, containing $n$ and $n_{test}$ observations respectively.  Write $N=n+n_{test}$ and $n=\lfloor g\times N\rfloor$, and let  $ g\in\{40\%, 50\%, 60\%, 70\%, 80\%\}$. We repeat the process of splitting the data 200 times, and calculate NHL and ER on that basis. The candidate models or base learners are generated in the same manner as in the simulation analysis. Our computation of the SVMMA weights is based on a five-fold CV.

Our first data set, labelled as ``Sonar, Mines versus Rocks" \citep{Dua:2019}, is downloadable from \url{https://archive.ics.uci.edu/ml/datasets/Connectionist+Bench+(Sonar,+Mines+vs.+Rocks)}. The data contain 60 features and 208 observations. Our response variable is a binary variable that takes on 1 or 0 depending on whether the object is rock or mine.  The features are signal information from sonars.
More description can be found in \cite{Sonardata}.  We set the number of candidate models to 20 based on the learning curve in Figure \ref{fig:Sonar_learn_curves}.

The second data set, labelled as ``Ionosphere", is downloadable from \url{https://archive.ics.uci.edu/ml/datasets/ionosphere}. The sample size is 351 with 33 features. The response is a binary variable that takes on 1 if the radar returns show evidence of some type of structure in the ionosphere, and 0 otherwise. The features are the information about electrons recorded by radars.  More descriptions can be found in \cite{Ionosphere1989}. Judging from the behaviour of the learning curve in Figure \ref{fig:Ionosphere_learn_curves},  we set the number of candidate models to 30.

The third data set labelled as ``LSVT Voice Rehabilitation", is downloadable from  \url{https://archive.ics.uci.edu/ml/datasets/LSVT+Voice+Rehabilitation}.  The sample size is 126 with 309 features. The response is a binary variable equal to 1 if the subject is diagnosed with Parkinson's disease, and 0 otherwise.
The features are the biomedical speech signals.  More description can be found in \cite{2015Objective}.
 We set the number of candidate models to 30 according to the learning curve in Figure \ref{fig:LSVT_learn_curves}.

The learning abilities of model averaging and ensemble learning methods are compared in Figures \ref{fig:Sonar_learn_curves}, \ref{fig:Ionosphere_learn_curves} and \ref{fig:LSVT_learn_curves}, where the $x$- and $y$-axes represent the number of candidate models and the error rates respectively. In all cases, the sizes of the training and test samples are 80\% and 20\% of the full sample size respectively.  As is expected, for a given method, the ER associated with the training data is typically smaller than that associated with the test data.    The learning ability of a method is often evaluated in terms of the stability of the ER it produces under the test sample - a method is deemed to be good if it leads to a small ER that stablises at a small value of the number of candidate models (i.e, $x$-axis). From that point of view, model averaging clearly outperforms the two ensemble learning methods, as revealed in the figures. One distinct advantage of SVMMA compared with bagging and adaboosting, is that it requires fewer base learners to achieve the same results.

Figures \ref{fig:Sonar}, \ref{fig:Ionosphere} and \ref{fig:LSVT} provide a comparison of the SVMMA with the SCH, SCL, SVMICH, SVMICL and UNIF methods in terms of NHL and ER in the test sample when the size of training sample varies between 40\% and 80\% of the total sample size.  The results show that SVMMA is frequently the best performer in the pool with respect to both yardsticks.
In addition, the fact that the curves of NHL of SVMMA are very close to 1,  corroborates Theorem \ref{thm:optimality} because NHL is an estimate of $R_n(\hat{\w})/\inf_{\w\in\calW} R_n(\w)$. When NHL is near 1, it implies $R_n(\hat{\w})/\inf_{\w\in\calW} R_n(\w)\to 1$.  Similar to the results of the simulation analysis, the SCL and SCH model averaging methods invariably deliver superior performance to their model selection counterparts, the SVMICL and SVMICH methods.

\section{Concluding remarks}
This paper develops a model averaging method based on cross-validation for SVM.  We provided a weight choice criterion and showed that the resulting SVMMA estimator is asymptotically
optimal in the sense of achieving the lowest hinge loss among all feasible models. We also developed a model screening method based on $L_1$ penalty.
Our simulation and real data analysis shows that the proposed model average estimator performs well, compared with several other selection, averaging and ensemble learning methods. Work in progress by the authors develops optimal model averaging strategies for multi-kernel SVM models, which involve the choice of kernel as another layer of uncertainty.

\begin{figure}[h]
  \centering
  \includegraphics[width=0.8\textwidth]{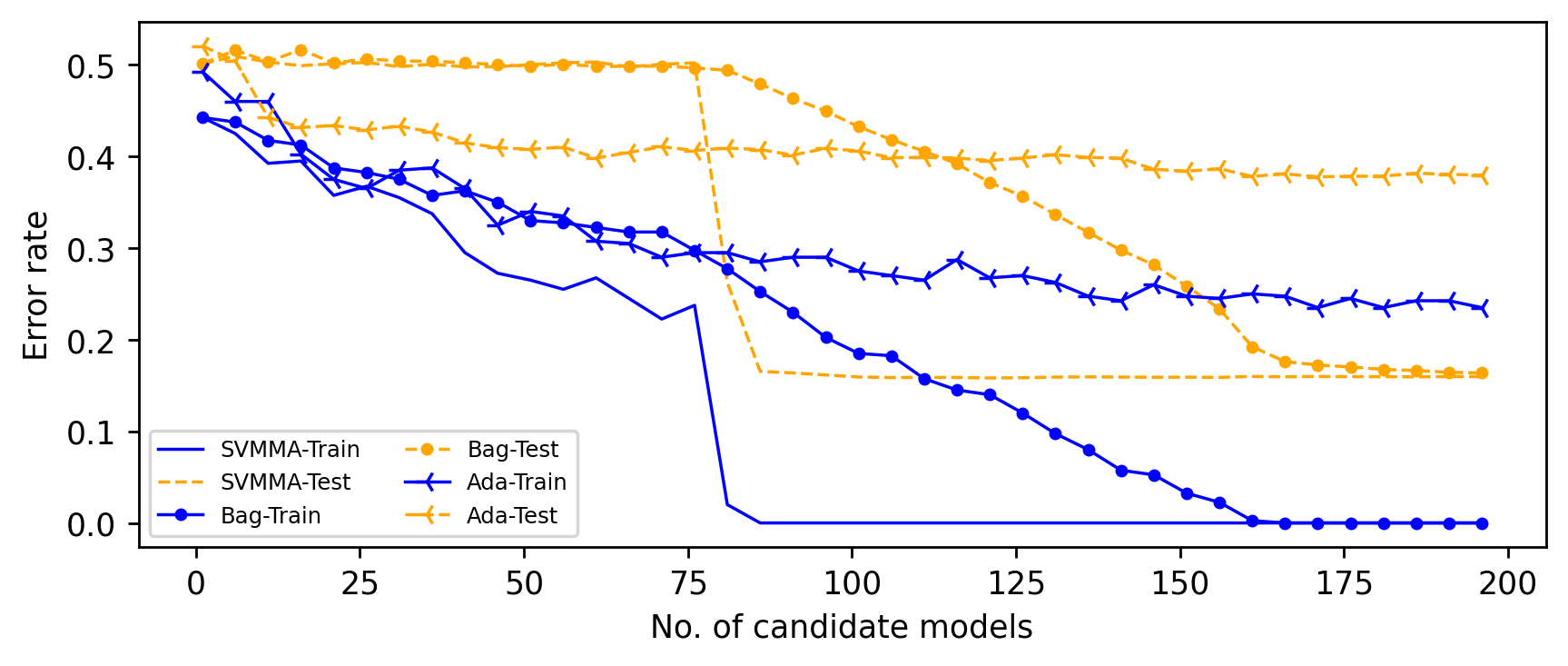}
  \caption{Learning curves under DGP1 and Scenario S1; $n$=400, $p$=1000.}\label{fig:Scenario11_learn_curves}
\end{figure}

\begin{figure}[h]
  \centering
  \includegraphics[width=0.8\textwidth]{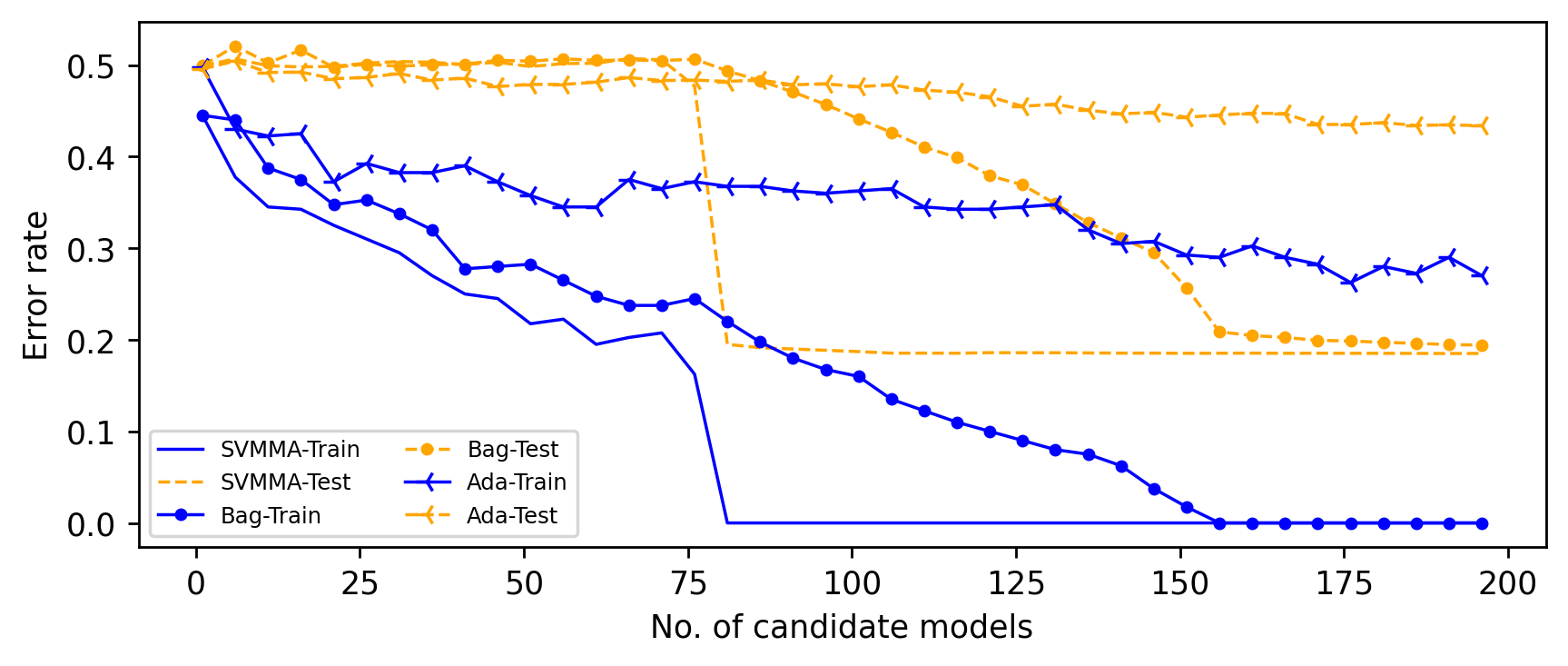}
  \caption{Learning curves under DGP2 and Scenario S1; $n$=400, $p$=1000.}\label{fig:Scenario12_learn_curves}

\end{figure}

\begin{figure}[h]
\centering
\begin{subfigure}{0.4\textwidth}
\includegraphics[width=\textwidth]{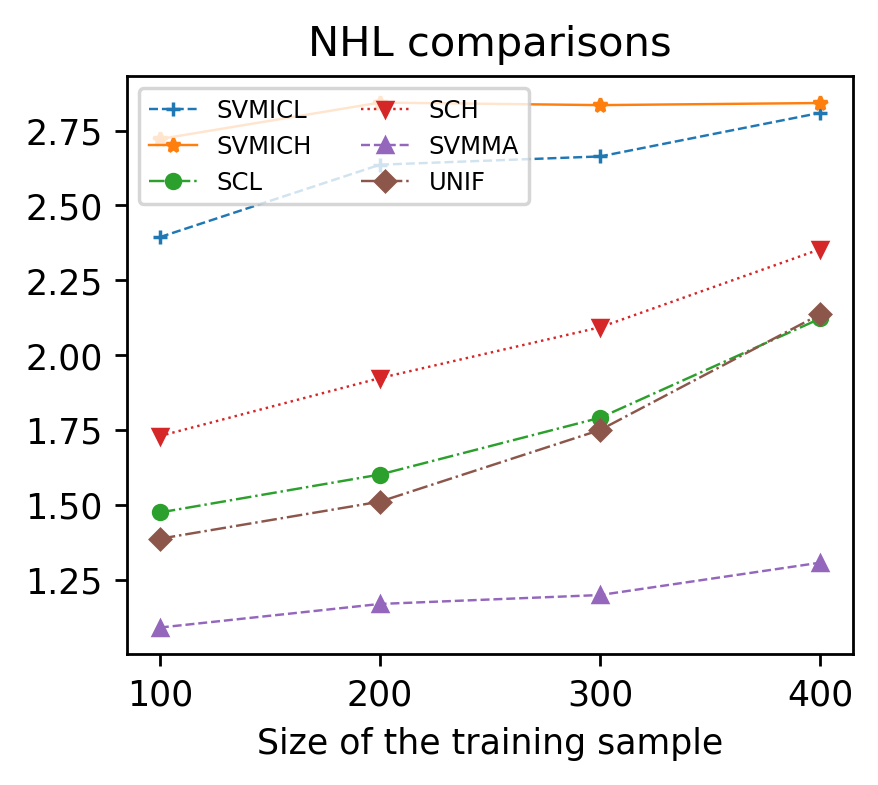}
\caption{DGP1}
\end{subfigure}
\qquad
\begin{subfigure}{0.4\textwidth}
\includegraphics[width=\textwidth]{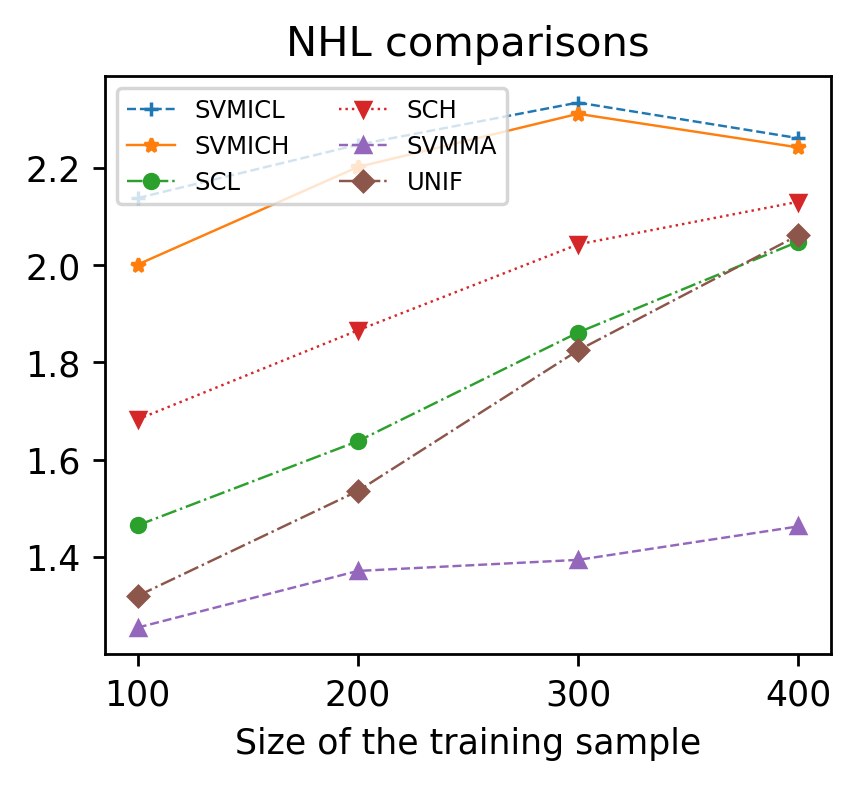}
\caption{DGP2}
\end{subfigure}
\caption{NHL under Scenario S1.}
\label{fig:miss-NHL}
\end{figure}

\begin{figure}[h]
\centering
\begin{subfigure}{0.4\textwidth}
\includegraphics[width=\textwidth]{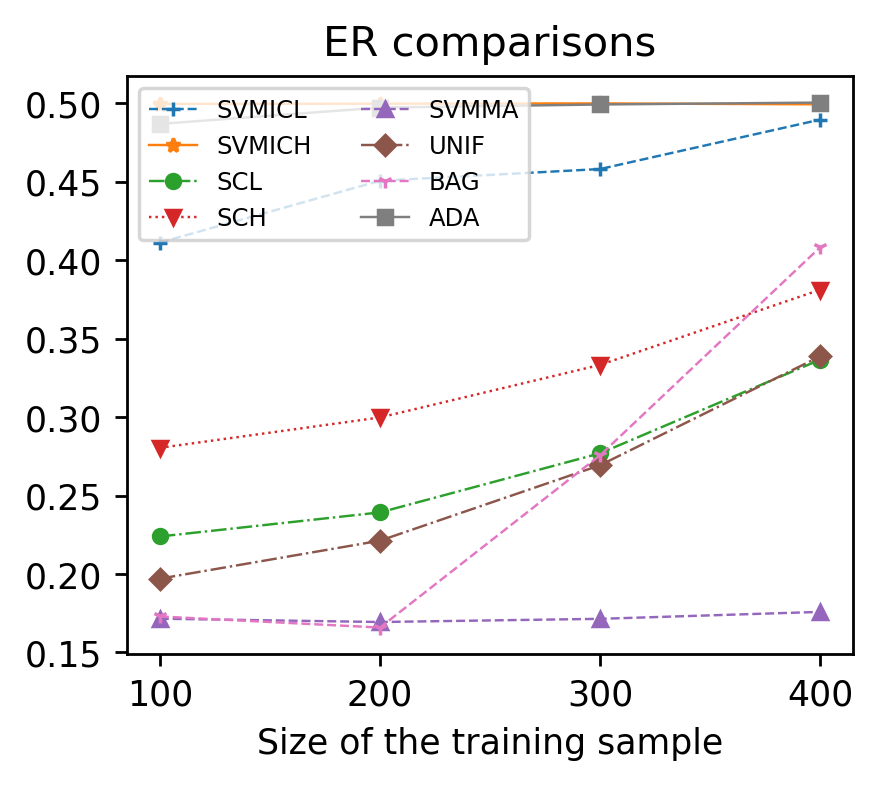}
\caption{DGP1}
\end{subfigure}
\qquad
\begin{subfigure}{0.4\textwidth}
\includegraphics[width=\textwidth]{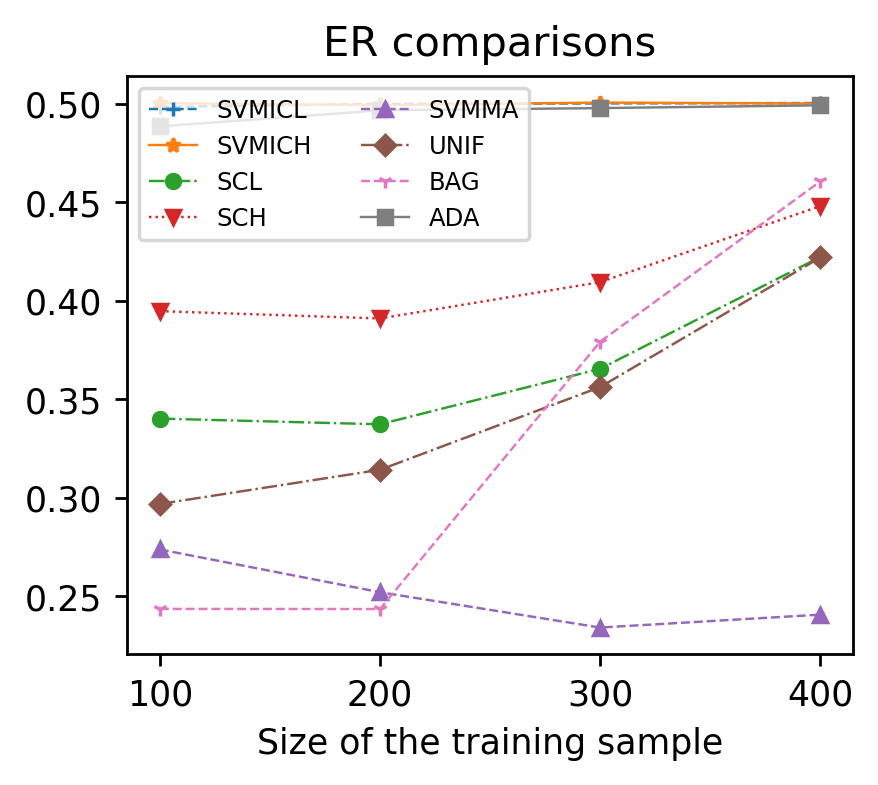}
\caption{DGP2}
\end{subfigure}
\caption{ER under Scenario S1.}
\label{fig:miss-ER}
\end{figure}

\begin{figure}[h]
\centering
\begin{subfigure}{0.4\textwidth}
\includegraphics[width=\textwidth]{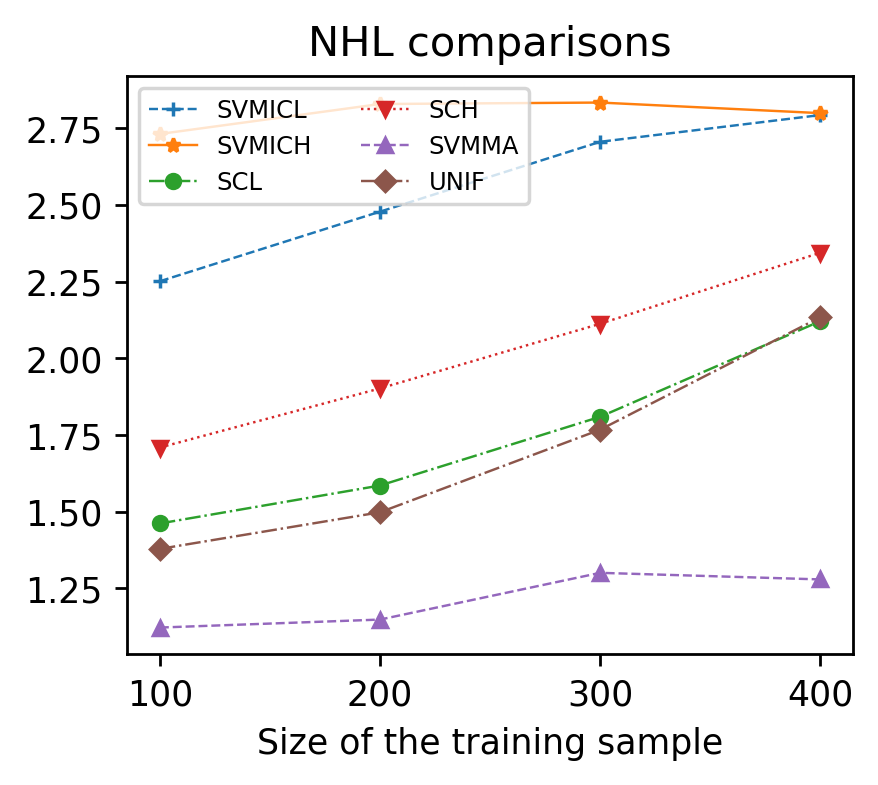}
\caption{DGP1}
\end{subfigure}
\qquad
\begin{subfigure}{0.4\textwidth}
\includegraphics[width=\textwidth]{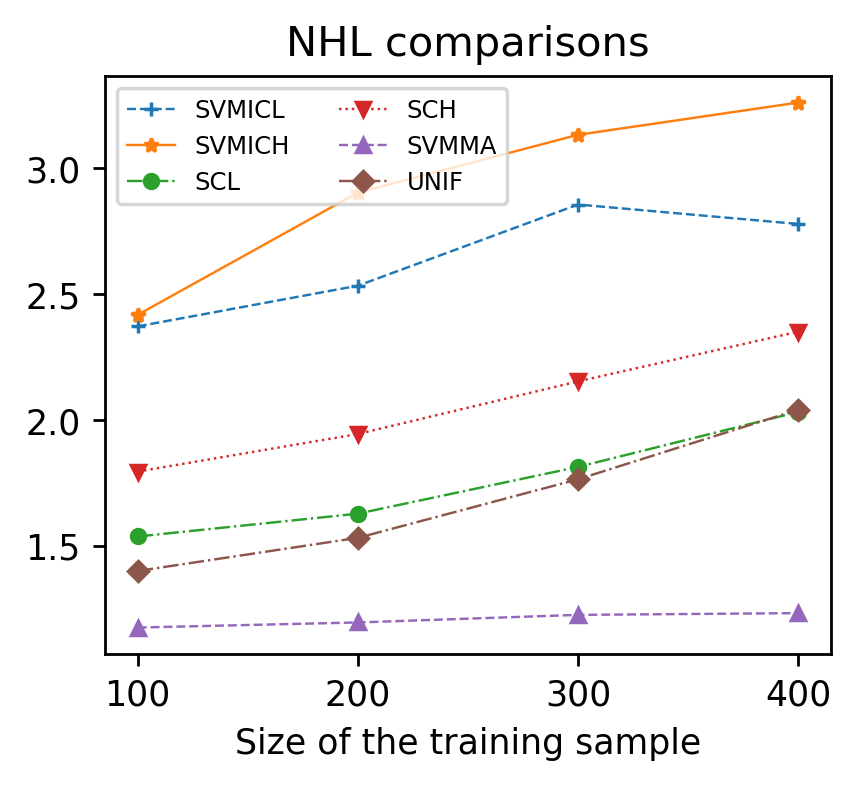}
\caption{DGP2}
\end{subfigure}
\caption{NHL under Scenario S2.}
\label{fig:correct-NHL}
\end{figure}

\begin{figure}[h]
\centering
\begin{subfigure}{0.4\textwidth}
\includegraphics[width=\textwidth]{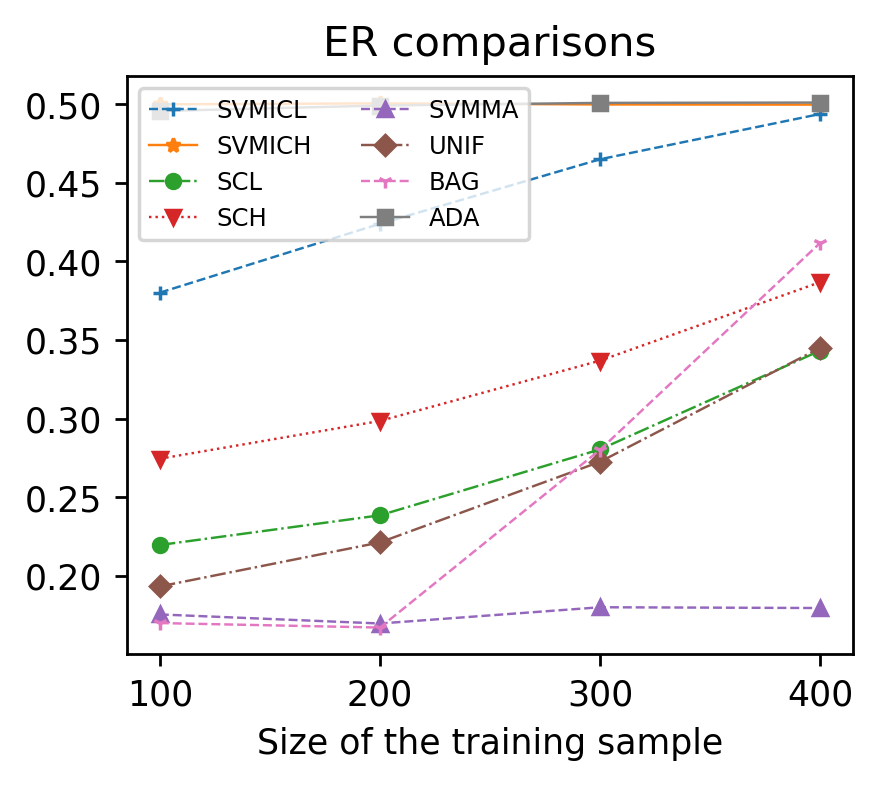}
\caption{DGP1}
\end{subfigure}
\qquad
\begin{subfigure}{0.4\textwidth}
\includegraphics[width=\textwidth]{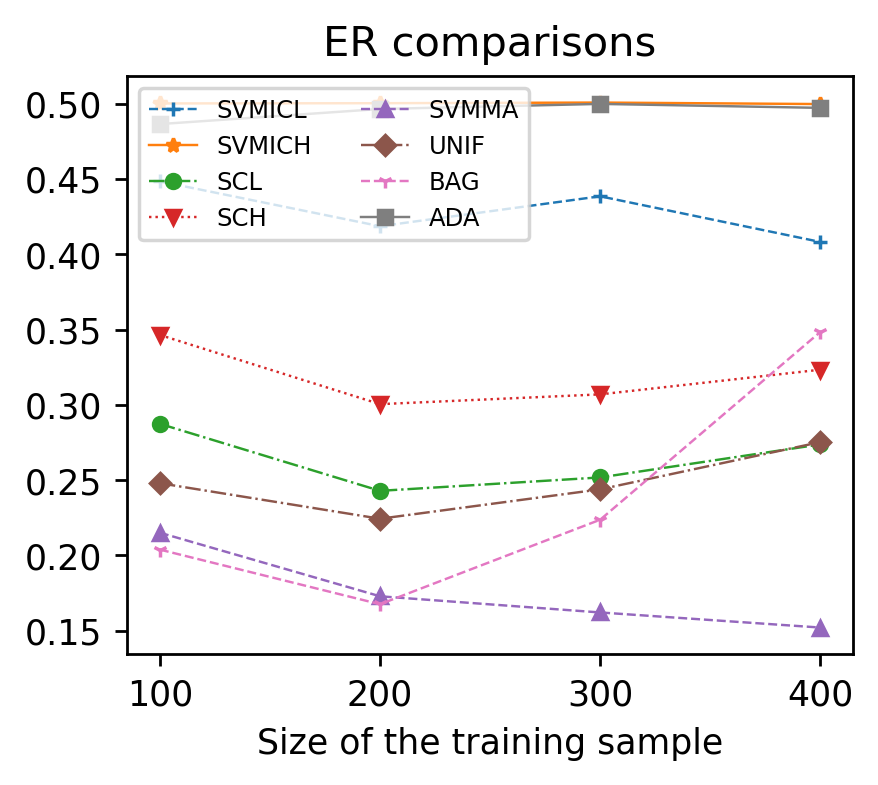}
\caption{DGP2}
\end{subfigure}
\caption{ER under Scenario S2.}
\label{fig:correct-ER}
\end{figure}

\begin{figure}[h]
  \centering
  \includegraphics[width=0.8\textwidth]{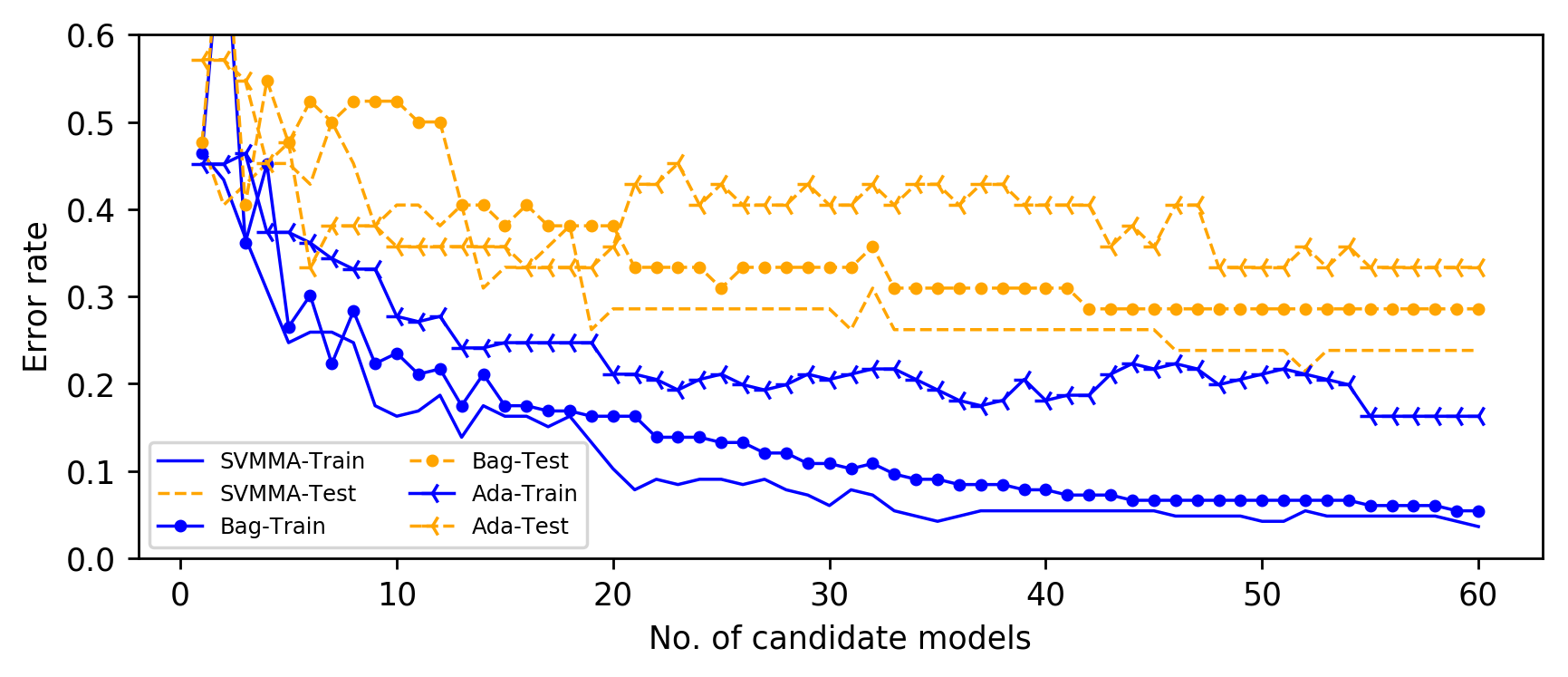}
  \caption{Learning curves for the Sonar data example; Sizes of the training and test samples are $\lfloor 80\% N\rfloor$ and $\lfloor 20\% N\rfloor$  respectively. }\label{fig:Sonar_learn_curves}
\end{figure}

\begin{figure}[h]
\centering
\begin{subfigure}{0.4\textwidth}
\includegraphics[width=\textwidth]{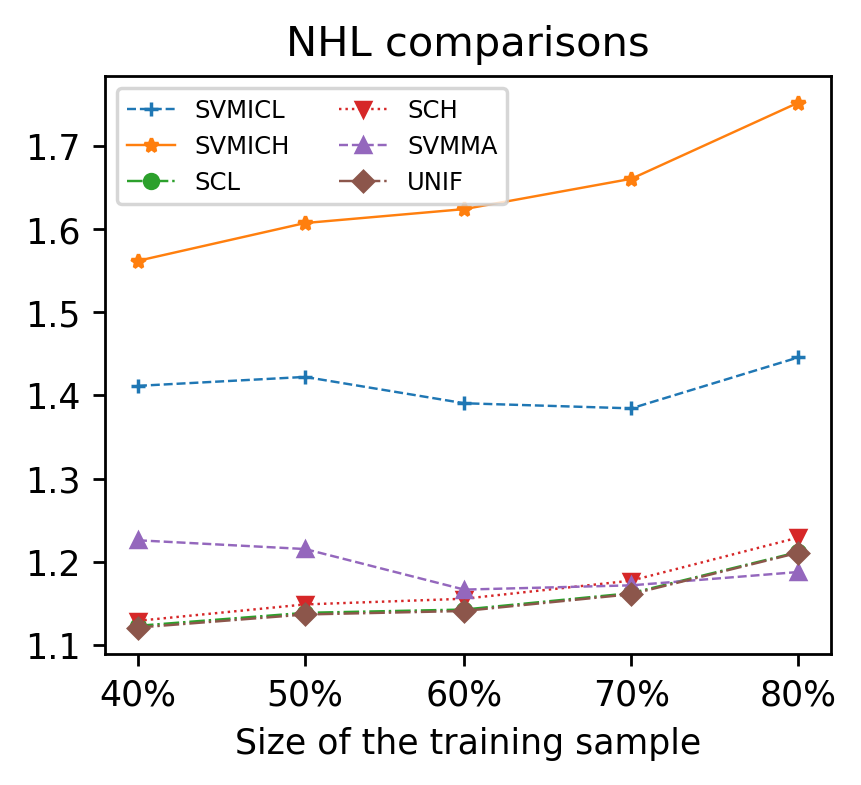}
\end{subfigure}
\qquad
\begin{subfigure}{0.4\textwidth}
\includegraphics[width=\textwidth]{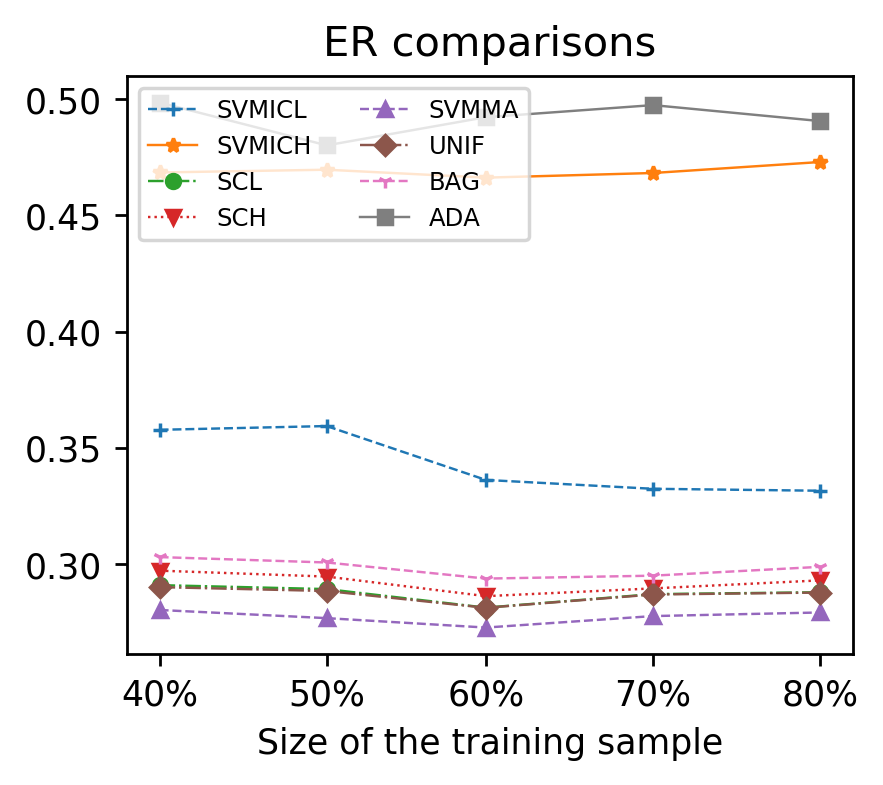}
\end{subfigure}
\caption{Comparisons for the Sonar data example.}
\label{fig:Sonar}
\end{figure}

\begin{figure}[h]
  \centering
  \includegraphics[width=0.8\textwidth]{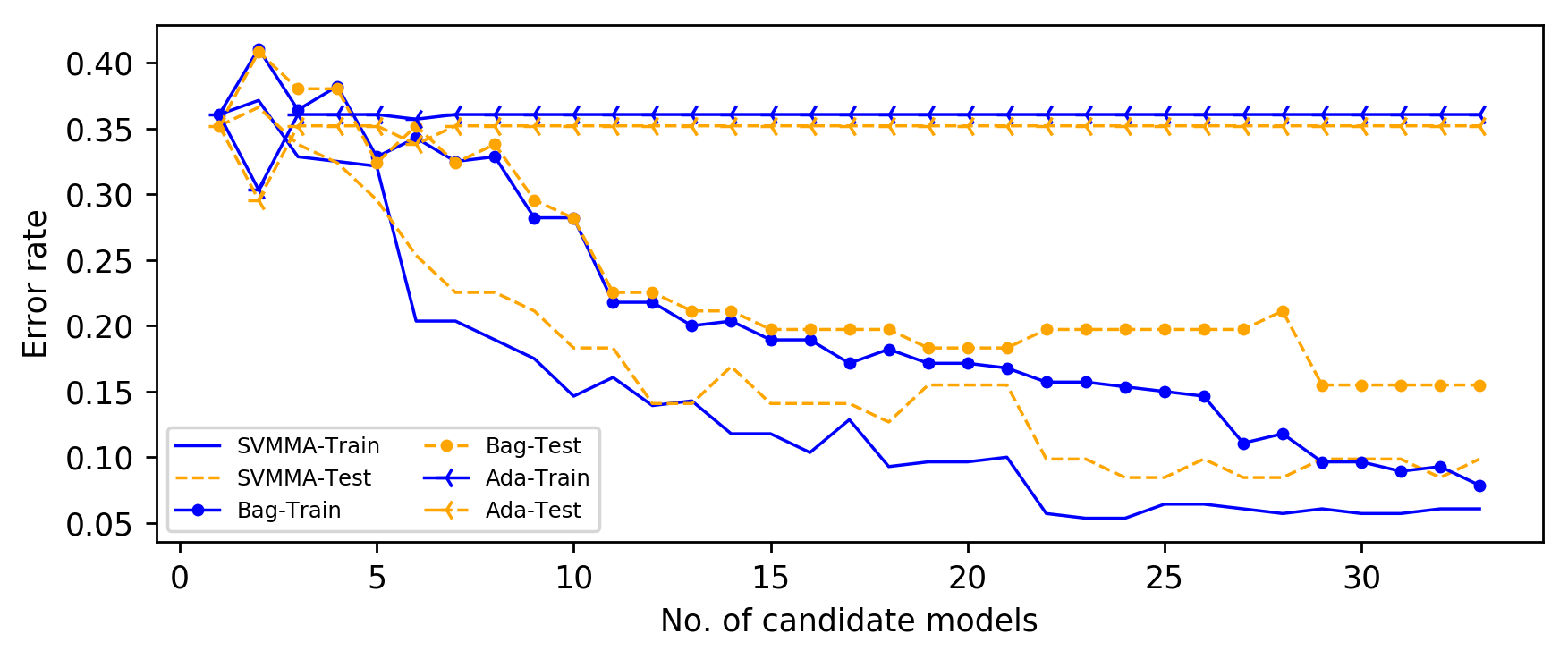}
  \caption{Learning curves for the Ionosphere data example; Sizes of the training and test samples are  $\lfloor 80\%N\rfloor$ and  $\lfloor 20\%N\rfloor$ respectively.}
  \label{fig:Ionosphere_learn_curves}
\end{figure}

\begin{figure}[h]
\centering
\begin{subfigure}{0.4\textwidth}
\includegraphics[width=\textwidth]{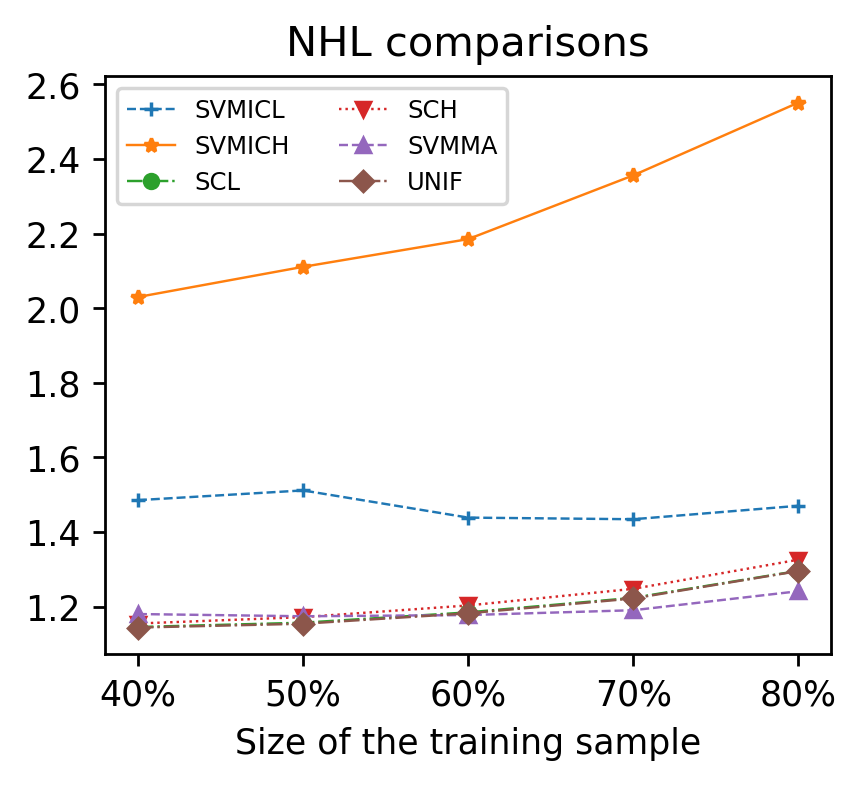}
\end{subfigure}
\qquad
\begin{subfigure}{0.4\textwidth}
\includegraphics[width=\textwidth]{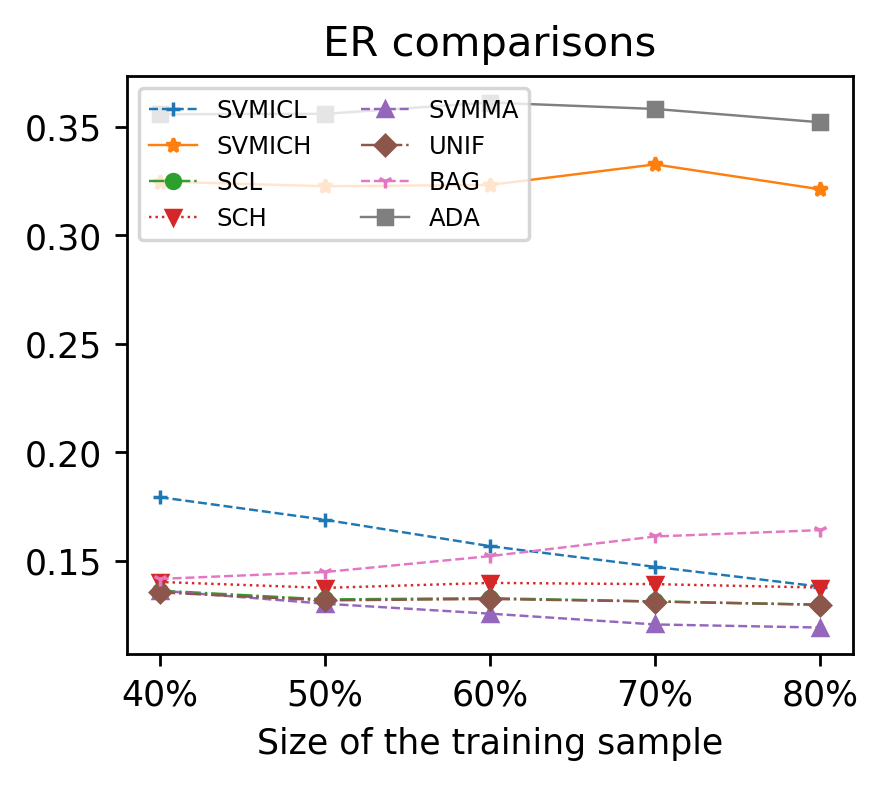}
\end{subfigure}
\caption{Comparisons for the Ionosphere data example.}
\label{fig:Ionosphere}
\end{figure}

\begin{figure}[h]
  \centering
  \includegraphics[width=0.8\textwidth]{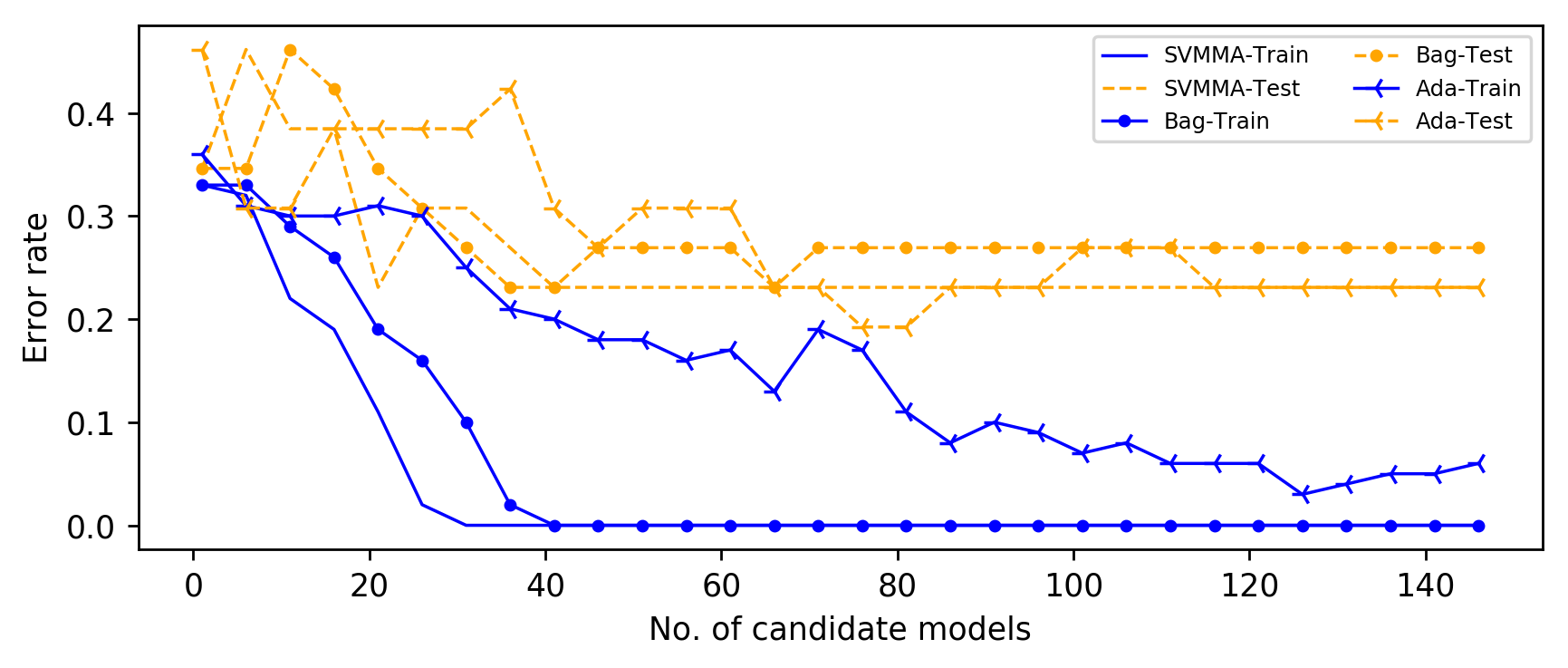}
  \caption{Learning curves for the LSVT data example; Sizes of the training and test samples are  $\lfloor 80\%N\rfloor$ and $\lfloor20\%N\rfloor$  respectively.}\label{fig:LSVT_learn_curves}
\end{figure}

\begin{figure}[h]
\centering
\begin{subfigure}{0.4\textwidth}
\includegraphics[width=\textwidth]{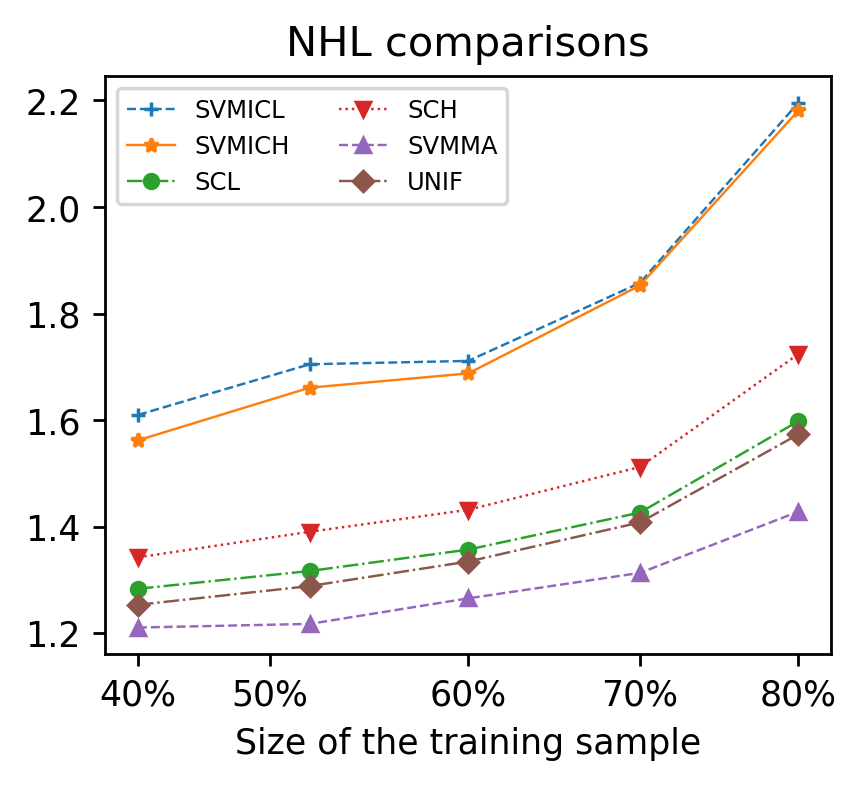}
\end{subfigure}
\qquad
\begin{subfigure}{0.4\textwidth}
\includegraphics[width=\textwidth]{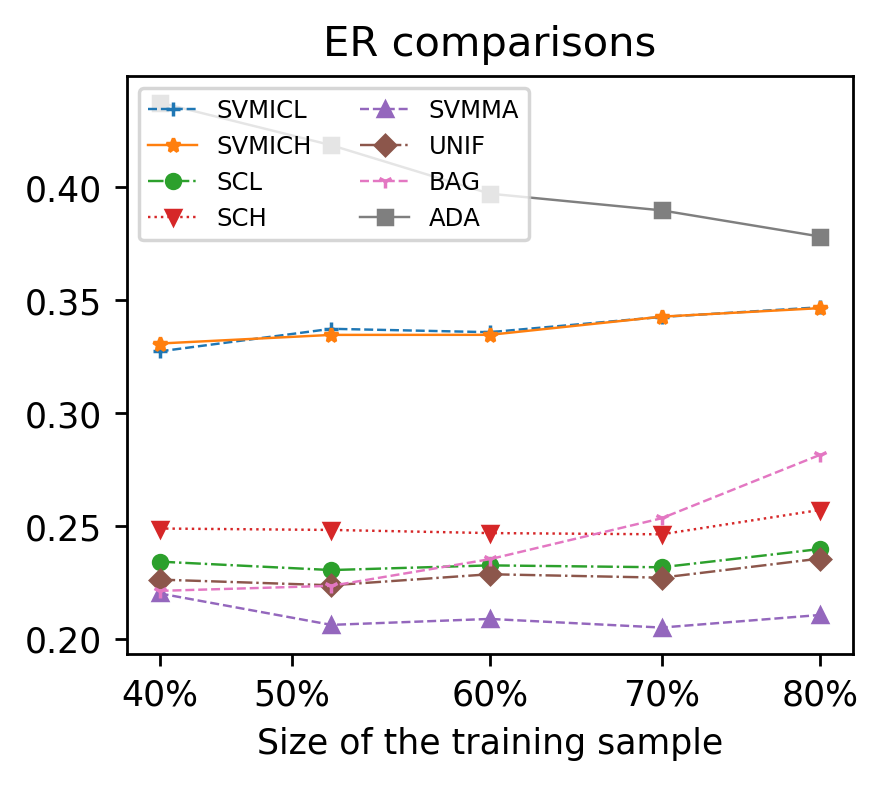}
\end{subfigure}
\caption{Comparisons for the LSVT data example.}
\label{fig:LSVT}
\end{figure}




\clearpage

\appendix
\section{Appendix}
This section provides the proof of Theorem \ref{thm:optimality}. All limiting processes below correspond to $n\to\infty$ unless stated otherwise.

\subsection{Proof of Lemma \ref{lem:consist}}
\begin{proof}
Part of this proof follows from \cite{Zhang2016aconsistent}, but there are some differences and the conclusion is also different from that of \cite{Zhang2016aconsistent}.

We will prove \eqref{eq:8} first.
Recall that $\hbeta_\s=\arg\min_{\bbeta_\s}\{n^{-1}\sumin (1-y_i\x_{\s,i}\tt\bbeta_\s)+2^{-1}\lambda_n\|\bbeta_\s^+\|^2\}$. We will show that, for any $0<\eta<1$, there exist a large constant $\triangle>0$ and an integer $N$ such that when $n>N$, we have

\begin{align}
\Pr\left\{\infs \inf_{\|\bu_\s\|=\triangle} \left\{ l_s\left(\sbeta_\s+\sqrt{n^{-1}\pmax\log(\pmax)}\bu_\s\right)-l_s(\sbeta_\s)\right\}>0 \right\}>1-\eta,\label{eq:191}
\end{align}
where $\bu_\s\in\mathcal{R}^{p_s}$ and $l_s(\bbeta_\s)=n^{-1}\sumin (1-y_i\x_{\s,i}\tt\bbeta_\s)_++2^{-1}\lambda_n\|\bbeta_\s^+\|^2$. As the hinge loss is convex, this implies that with probability $1-\eta$, $\sups \|\hbeta_\s-\sbeta_\s\|\leq \triangle \sqrt{n^{-1}\pmax\log(\pmax)} $.  Hence equation \eqref{eq:8} in Lemma \ref{lem:consist} holds.

Note that $l_s\left(\sbeta_\s+\sqrt{n^{-1}\pmax\log(\pmax)}\bu_\s \right)-l_s\left(\sbeta_\s\right)$ can be expressed as
\begin{align}
&\quad l_s\left(\sbeta_\s+\sqrt{n^{-1}\pmax\log(\pmax)}\bu_\s \right)-l_s\left(\sbeta_\s\right)\notag\\
&=n^{-1} \sumin \left\{\left(1-y_i\x_{\s,i}\tt (\sbeta_\s+\sqrt{n^{-1}\pmax\log(\pmax)}\bu_\s) \right)_+-\left(1-y_i\x_{\s,i}\tt\sbeta_\s\right)_+ \right\}\notag\\
&\quad + 2^{-1}\lambda_n \left\|\spbeta_\s +\sqrt{n^{-1}\pmax\log(\pmax)}\bu_\s^+\right\|^2-2^{-1}\lambda_n \|\spbeta_\s\|^2.\label{eq:202}
\end{align}
It is readily shown that
\begin{align}
  &\quad\sups\sup_{\|\bu_\s\|=\triangle}\left|\left\|\spbeta_\s+\sqrt{n^{-1}\pmax\log(\pmax)}\bu_\s^+ \right\|^2-\left\|\spbeta_\s\right\|^2 \right|\notag\\
  &\leq \sups\sup_{\|\bu_\s\|=\triangle}\left(\left\|\spbeta_\s+\sqrt{n^{-1}\pmax\log(\pmax)}\bu_\s^+ \right\|+\left\|\spbeta_\s\right\| \right)\left\|\sqrt{n^{-1}\pmax\log(\pmax)}\bu_\s^+ \right\|\notag\\
  &\leq2\triangle C_2\pmax \sqrt{n^{-1}\log(\pmax)}+\triangle n^{-1}\pmax\log(\pmax)\notag\\
  & =O(\triangle \pmax \sqrt{n^{-1}\log(\pmax)}),\label{eq:penalty}
 \end{align}
   where the last inequality is obtained from Condition~\ref{con:1ii2}. Hence the order  of difference of penalty terms in \eqref{eq:202} is $O(\triangle \lambda_n\pmax \sqrt{n^{-1}\log(\pmax)})$.

Denote
\begin{align*}
g_{s,i}(\bu_\s)&=\left(1-y_i\x_{\s,i}\tt(\sbeta_\s+\sqrt{n^{-1} \pmax\log(\pmax)}\bu_\s) \right)_+-\left(1-y_i\x_{\s,i}\tt \sbeta_\s\right)_+\\
&\quad +\sqrt{n^{-1}\pmax\log(\pmax)}y_i\x_{\s,i}\tt\bu_\s \mathbf{1}\left(1-y_i\x_{\s,i}\tt\sbeta_\s\geq 0 \right)\notag\\
&\quad -\Exp\left[\left(1-y_i\x_{\s,i}\tt(\sbeta_\s+\sqrt{n^{-1}\pmax\log(\pmax)}\bu_\s)\right)_+ \right]+\Exp\left[\left(1-y_i\x_{\s,i}\tt\sbeta_\s\right)_+ \right].
\end{align*}
It can be verified that $\Exp[g_{\s,i}(\bu)]=0$, $s=1,2,...,S_n$ by the definition of $\sbeta_\s$ and $\J_\s(\sbeta_\s)=0$. Note that \eqref{eq:202} can be further decomposed as
\begin{align*}
n^{-1}\sumin \left\{\left(1-y_i\x_{\s,i}\tt (\sbeta_\s+\sqrt{n^{-1}\pmax\log(\pmax)}\bu_\s) \right)_+-\left(1-y_i\x_{\s,i}\tt\sbeta_\s\right)_+ \right\}=n^{-1} (A_{s,n}+B_{s,n} ),
\end{align*}
where
\begin{align*}
  A_{s,n}=\sumin g_{s,i}(\bu_\s)
\end{align*}
and
\begin{align}
  B_{s,n}&=\sumin \Big[ -\sqrt{n^{-1}\pmax\log(\pmax)}y_i \x_{\s,i}\tt\bu_\s \mathbf{1}\left(1-y_i\x_{\s,i}\tt\sbeta_\s\geq 0 \right)\notag\\
  &\quad+ \Exp\left\{\left(1-y_i\x_{\s,i}\tt(\sbeta_\s+\sqrt{n^{-1}\pmax\log(\pmax)}\bu_\s) \right)_+ \right\}-\Exp\left\{\left(1-y_i\x_{\s,i}\tt\sbeta_\s \right)_+\right\}\Big].\label{eq:211}
\end{align}

The remainder of the proof consists of three steps. In Step 1, we demonstrate that
\begin{align} \label{step1}
\sups \sup_{\|\bu_\s\|=\triangle}|A_{s,n}|=\triangle^{3/2}\pmax o_{p}(1).
\end{align}
In Step 2, it is shown that $\infs \inf_{\|\bu_\s\|=\triangle}B_{s,n}$ dominates the terms of order $\triangle^{3/2}\pmax o_{p}(1)$ and is larger than zero. In Step 3, we use the results from the previous steps to prove \eqref{eq:191}.

Step 1: We use the covering number introduced by \cite{vaart1996weakconvergence} to prove the uniform rate in (\ref{step1}). It suffices to show, for any $\epsilon>0$, that
\begin{align}
\Pr\left(\sups \sup_{\|\bu_\s\|=\triangle} p_s^{-1}\left|\sumin g_{s,i}(\bu_\s) \right|>\triangle^{3/2}\epsilon \right)\to 0.\label{eq:222}
\end{align}
 Note that the hinge loss satisfies the Lipschitz condition and  $\supi \|\x_{\s,i}\|\leq C_1\sqrt{p_s}$, $\supi\Exp\|\x_{\s,i}\|\leq C_1\sqrt{p_s}$ from Condition~\ref{con:1ii}.
 It is readily shown that
\begin{align}
 |g_{s,i}(\bu_\s)|
  &\leq 3\triangle \sqrt{n^{-1}\pmax\log(\pmax)}  \max\left\{\supi\|\x_{\s,i}\|,\supi\Exp\|\x_{\s,i}\|\right\}\notag\\
  &\leq 3C_1\triangle\pmax \sqrt{n^{-1}\log(\pmax)}\label{eq:psg}
\end{align}
and thus $\sups\sup_{\|\bu_\s\|=\triangle}p_s^{-1}|g_{s,i}(\bu_\s)|=o(1)$ by Condition~\ref{con:1v}.
By Lemma 2.5 of \cite{vande2000Empirical}, the ball $\{\bu_\s:\|\bu_\s\|\leq \triangle\}$ in $\mathcal{R}^{p_s+1}$ can be covered by $N_s$ balls with radius $\zeta_s$, where $N_s\leq \{(4\triangle +\zeta_s)/\zeta_s\}^{p_s+1}$. Denote $\bu_\s^{1},..,\bu_\s^{N}$ as the centers of the $N_s$ balls, let $\zeta_s=(nM_1)^{-1} p_s$ (for some large constant $M_1>0$ ) and denote $\mathcal{U}_s^{k}=\{\bu_\s: \|\bu_\s-\bu_\s^{k}\|\leq \zeta_s \& \|\bu_\s\|=\triangle\}$. For any $\epsilon>0$, we have
\begin{align}
&\quad\sups\max_{1\leq k\leq N_s}  \sup_{\bu_\s\in\mathcal{U}_s^{(k)}} p_s^{-1}\left| \sumin g_{s,i}(\bu_\s)-\sumin g_{s,i}(\bu_\s^k)\right|\notag\\
&\leq \sups\max_{1\leq k\leq N_s}  \sup_{\bu_\s\in\mathcal{U}_s^{k}} p_s^{-1}\sumin\left|  g_{s,i}(\bu_\s)- g_{s,i}(\bu_\s^k)\right|\notag\\
&\leq \sups\max_{1\leq k\leq N_s}  \sup_{\bu_\s\in\mathcal{U}_s^{k}} np_s^{-1}\Big\{2\sqrt{n^{-1}\pmax\log(\pmax)}\|\x_{\s,i}\| \|\bu_\s-\bu_\s^{k}\|\notag\\
&\quad +\sqrt{n^{-1}\pmax\log(\pmax)}\|\bu_\s-\bu_\s^{k}\|\Exp\|\x_{\s,i}\| \Big\} \notag\\
&\leq\sups 3\triangle n p_s^{-1} \sqrt{n^{-1}\pmax\log(\pmax)}\max\left\{\supi \|\x_{\s,i}\|,\supi\Exp\|\x_{\s,i}\|\right\}\zeta_s\notag\\
&\leq 3C_1 M_1^{-1}\triangle \pmax \sqrt{n^{-1}\log(\pmax)}\notag\\
&=o( \triangle^{3/2}p_{\min}\epsilon/2),\label{eq:11}
\end{align}
where the last inequality arises from Condition \ref{con:1v}. From \eqref{eq:11}, it can be shown that
\begin{align}
&\Pr\left(\sups \sup_{\|\bu_\s\|=\triangle} p_s^{-1}\left |\sumin g_{s,i}(\bu_\s) \right|>\triangle^{3/2} \epsilon \right)\notag\\
&\leq \Pr\Bigg(\sups \max_{1\leq k\leq N_s} \sup_{\bu_\s\in\mathcal{U}_s^{(k)}}p_s^{-1}\left|\sumin g_{s,i}(\bu_\s)-\sumin g_{s,i}(\bu_\s^{k})\right|\notag\\
&\quad+\sups \max_{1\leq k \leq N_s}p_s^{-1}\left|\sumin g_{s,i}(\bu_\s^{k}) \right|>\triangle^{3/2}\epsilon \Bigg)\notag\\
&\leq \Pr\Bigg(\sups \max_{1\leq k\leq N_s} \sup_{\bu_\s\in\mathcal{U}_s^{(k)}}\left|\sumin g_{s,i}(\bu_\s)-\sumin g_{s,i}(\bu_\s^{k})\right|>\triangle^{3/2}p_{\min}\epsilon/2\Bigg)\notag\\
&\quad+\sums \sum_{k=1}^{N_s}\Pr\left(  \left|\sumin g_{s,i}(\bu_\s^{k}) \right|>\triangle^{3/2}p_s\epsilon/2 \right)\notag\\
&= \sums \sum_{k=1}^{N_s}\Pr\left(  \left|\sumin g_{s,i}(\bu_\s^{k}) \right|>\triangle^{3/2}p_s\epsilon/2 \right)+o(1)\label{eq:242}
\end{align}
and $\sumin g_{s,i}(\bu_\s^{(k)})$ is the sum of independent zero-mean random variables.

By the bounded conditional density, under Conditions \ref{con:1i} and \ref{con:1iv}, recognising that\\ $\supi \|\x_{\s,i}\|\leq C_1\sqrt{p_s}$, we have
\begin{align}
&\quad \Pr\left(|1-y_i\x_{\s,i}\tt\sbeta_\s|\leq \sqrt{n^{-1}\pmax\log(\pmax)} \supi\|\x_{\s,i}\|\triangle \right)\notag\\
&=\Pr\Big(\pm1-\sqrt{n^{-1}\pmax\log(\pmax)} \supi\|\x_{\s,i}\|\triangle\leq\x_{\s,i}\tt\sbeta_\s\leq \notag\\
 &\quad\sqrt{n^{-1}\pmax\log(\pmax)} \supi\|\x_{\s,i}\|\triangle\pm1\Big| y_i=\pm 1\Big)\notag\\
&\leq 2C_3\sqrt{n^{-1}\pmax \log(\pmax)}\supi \|\x_{\s,i}\|\triangle\notag\\
&\leq 2\triangle C_1 C_3\sqrt{n^{-1} p_s\pmax\log(\pmax)}.
\label{eq:281}
\end{align}
Note that when $1-y_i\x_{\s,i}\tt\sbeta_\s<\sqrt{n^{-1}\pmax\log(\pmax)}y_i\x_{\s,i}\tt\bu_\s$ and $\sqrt{n^{-1}\pmax\log(\pmax)}\\y_i\x_{\s,i}\tt\bu_\s<0$, or when $1-y_i\x_{\s,i}\tt\sbeta_\s>\sqrt{n^{-1}\pmax\log(\pmax)}y_i\x_{\s,i}\tt\bu_\s$ and $\sqrt{n^{-1}\pmax\log(\pmax)}\\y_i\x_{\s,i}\tt\bu_\s>0$,
\begin{align}
&\left(1-y_i\x_{\s,i}\tt(\sbeta_\s+\sqrt{n^{-1}\pmax\log(\pmax)}\bu_\s) \right)_+-(1-y_i\x_{\s,i}\tt\sbeta_\s)_+\notag\\
&\quad +\sqrt{n^{-1}\pmax\log(\pmax)} y_i\x_{\s,i}\tt\bu_\s\mathbf{1}(1-y_i\x_{\s,i}\tt\sbeta_\s\geq 0)=0.\label{eq:221}
\end{align}
Furthermore, equation \eqref{eq:221} holds when $|1-y_i\x_{\s,i}\tt\sbeta_\s|>\sqrt{n^{-1}\pmax\log(\pmax)}\supi \|\x_{\s,i}\|\triangle$ as $\sqrt{n^{-1}\pmax\log(\pmax)}\supi \|\x_{\s,i}\|\triangle>\left|\sqrt{n^{-1}\pmax\log(\pmax)}y_i\x_{\s,i}\tt\bu_\s\right|$. Hence we can write
{
\begin{align}
&\quad\sumin \Exp\{g^2_{s,i}(\bu_\s^k)\}\\
&\leq \sumin \Exp\Bigg[\Bigg\{\Big|\left(1-y_i\x_{\s,i}\tt(\sbeta_\s+\sqrt{n^{-1}\pmax\log(\pmax)}\bu_\s^{k}) \right)_+-\left(1-y_i\x_{\s,i}\tt \sbeta_\s\right)_+\Big|\notag\\
&\quad+\Big|\sqrt{n^{-1}\pmax\log(\pmax)}y_i\x_{\s,i}\tt\bu_\s^{k} \Big| \Bigg\}^2\mathbf{1}\Big(|1-y_i\x_{\s,i}\tt\sbeta_\s|\leq \sqrt{n^{-1}\pmax\log(\pmax)}\notag\\
&\quad\times\supi\|\x_{\s,i}\|\triangle\Big)\Bigg]\notag\\
&\leq \sumin \Exp \Bigg\{ \left(2\sqrt{n^{-1}\pmax\log(\pmax)}\x_{\s,i}\tt\bu_\s^k\right)^2\mathbf{1}\Big(|1-y_i\x_{\s,i}\tt\sbeta_\s|\leq \sqrt{n^{-1}\pmax\log(\pmax)}\notag\\
&\quad\times\supi\|\x_{\s,i}\|\triangle\Big)\Bigg\}\notag\\
&\leq
\left(2\sqrt{n^{-1}\pmax\log(\pmax)}\supi\|\x_{\s,i}\|\triangle\right)^2\sumin\Exp\Big(|1-y_i\x_{\s,i}\tt\sbeta_\s|\leq \sqrt{n^{-1}\pmax\log(\pmax)}\notag\\
&\quad\times\supi\|\x_{\s,i}\|\triangle\Big)\notag\\
&\leq  4 C_1^2\triangle^2 n^{-1} p_s \pmax \log(\pmax)\sumin\Exp\Big(|1-y_i\x_{\s,i}\tt\sbeta_\s|\leq \sqrt{n^{-1}\pmax\log(\pmax)}\supi\|\x_{\s,i}\|\triangle\Big)\notag\\
&\leq 4 C_1^2\triangle^2 n^{-1} p_s \pmax \log(\pmax)\times 2n\triangle C_1 C_3\sqrt{n^{-1} p_s\pmax\log(\pmax)}\notag\\
&=8\triangle^3 C_1^3C_3n^{-1/2}p_s^{3/2}p^{3/2}_{\max}\log^{3/2}(\pmax)
,\label{eq:22}
\end{align}
where the second-to-last inequality arises from $\supi \|\x_{\s,i}\|\leq C_1\sqrt{p_s}$ and the last inequality is from \eqref{eq:281}.
}
%
Finally, by Bernstein's inequality and recognising \eqref{eq:psg} and \eqref{eq:22}, we can write
\begin{align}
&\quad\sums\sum_{k=1}^{N_s}\Pr\left( \left|\sumin g_{s,i}(\bu^k) \right|>\triangle^{3/2}p_s\epsilon/2 \right)\notag\\
&\leq \sums\sum_{k=1}^{N_s}2\exp\left(-\frac{\triangle^{3}p_s^2\epsilon^2/4}{\sumin \Exp\{g_{s,i}^2(\bu^k)\}+3\triangle^{5/2} C_1p_s\pmax \sqrt{n^{-1}\log(\pmax)} \epsilon/2 }\right)\notag\\
&\leq \sums \left(\frac{4\triangle+(nM_1)^{-1}p_s}{(nM_1)^{-1}p_s} \right)^{p_s+1} \notag\\
&\quad\times\exp\left(-\frac{\triangle^{3}p_s^2\epsilon^2/4}{ 8\triangle^3 C_1^3 C_3 n^{-1/2}p_s^{3/2}p^{3/2}_{\max}\log^{3/2}(\pmax)+
3\triangle^{5/2} C_1p_s\pmax \sqrt{n^{-1}\log(\pmax)} \epsilon/2  }\right)\notag\\
&\leq S_n \left(\frac{4\triangle M_1n}{p_{\min}}+1 \right)^{\pmax+1}
\exp\left(-\frac{\triangle^{3}p_{\min}^{1/2}\epsilon^2/4}{ 16\triangle^3 C_1^3 C_3 n^{-1/2}p^{3/2}_{\max}\log^{3/2}(\pmax) }\right)\notag\\
&=O(1)\exp\left\{\log(S_n)+(\pmax+1)\log(4\triangle nM_1p^{-1}_{\min}+1)-64^{-1}C_1^{-3}C_3^{-1}\epsilon^2n^{1/2}p_{\min}^{1/2}p^{-3/2}_{\max}\log^{-3/2}(\pmax)\right\}\notag\\
&=o(1),\label{eq:301}
\end{align}
where the last equality is due to Condition~\ref{con:1v} { and
 $S_n=O\{\exp(n^{\tau})\}$ for $\tau\in (0, 1/2-3\kappa/2)$.}
 The proof of \eqref{eq:222} is complete by combining \eqref{eq:242} and \eqref{eq:301}.

Step 2: Let us rewrite $B_{s,n}$ as $B_{s,n}\equiv B_{s,n1}+B_{s,n2}$, where
\begin{align}
 B_{s,n1}&= -\sumin \sqrt{n^{-1}\pmax\log(\pmax)}y_i \x_{\s,i}\tt\bu_s \mathbf{1}\left(1-y_i\x_{\s,i}\tt\sbeta_\s\geq 0 \right),\notag\\
\textrm{and} \notag    \\
B_{s,n2}&= \Exp\left\{\left(1-y_i\x_{\s,i}\tt(\sbeta_\s+\sqrt{n^{-1}\pmax\log(\pmax)}\bu_\s) \right)_+ \right\}-\Exp\left\{\left(1-y_i\x_{\s,i}\tt\sbeta_\s \right)_+\right\}.\notag
\end{align}
To analyse $B_{s,n1}$, we observe that
\begin{align}
&\quad\left|\sumin  y_i\x_{\s,i}\tt\bu_\s\mathbf{1}\left(1-y_i\x_{\s,i}\tt\sbeta_\s\geq 0 \right) \right|\notag\\
&=\left|\sum_{j=0}^{p_s}\sumin y_i\x_{\s,ij}u_{\s,j}\mathbf{1}\left(1-y_i\x_{\s,i}\tt\sbeta_\s\geq 0 \right) \right|\notag\\
&\leq\sum_{j=0}^{p_s}|u_{\s,j}|\max_{0\leq j\leq p_s}\left|\sumin y_i x_{\s,ij}\mathbf{1}\left(1-y_i\x_{\s,i}\tt\sbeta_\s\geq 0 \right) \right|\notag\\
&\leq\sqrt{\sum_{j=0}^{p_s}u_{\s,j}^2}\sqrt{\sum_{j=0}^{p_s}1}\max_{0\leq j\leq p_s}\left|\sumin y_ix_{\s,ij}\mathbf{1}\left(1-y_i\x_{\s,i}\tt\sbeta_\s\geq 0 \right) \right|\notag\\
&\leq \sqrt{p_s+1} \triangle\max_{0\leq j\leq p_s}\left|
\sumin y_i x_{\s,ij}\mathbf{1}\left(1-y_i\x_{\s,i}\tt\sbeta_\s\geq 0 \right)\right|.\label{eq:251}
\end{align}
By the definition of $\J_s(\sbeta_\s)$, note that $\Exp\left[y_ix_{\s,ij}\mathbf{1}\left(1-y_i\x_{\s,i}\tt\sbeta_\s\geq 0 \right)\right]=0$ for $0\leq j \leq p_s$. By Lemma 14.24 in \cite{Buhlmann2011Statisticshigh} (the Nemirovski moment inequality),
\begin{align}
&\quad\Exp\left\{\max_{0\leq j\leq p_s}\left|\sumin y_ix_{\s,ij}\mathbf{1}\left(1-y_i\x_{\s,i}\tt\sbeta_\s\geq 0 \right) \right|\right\}\notag\\
&\leq \sqrt{8\log(2p_s+2)}\Exp\left(\max_{1\leq j\leq p_s+1}\sumin y_i^2x^2_{\s,ij}\right)^{1/2}\notag\\
&\leq\sqrt{ 8\log(2p_s+2)}\sqrt{nC_1^2}\notag\\
&=O(\sqrt{n\log(p_s)}),\label{eq:35}
\end{align}
where the last inequality is established by Condition \ref{con:1ii}.
Additionally, using Markov's inequality and by \eqref{eq:35}, we obtain
\begin{align}
\max_{0\leq j\leq p_s}\left|\sumin y_ix_{\s,ij}\mathbf{1}\left(1-y_i\x_{\s,i}\tt\sbeta_\s\geq 0 \right) \right|= O_{p}(\sqrt{n \log(p_s)}).\label{eq:261}
\end{align}
Combining \eqref{eq:251} and \eqref{eq:261}, we have
\begin{align}
&\quad\sups\sup_{\|\bu_\s\|=\triangle}|B_{s,n1}|\notag\\
&=\sups\sup_{\|\bu_\s\|=\triangle}\left|\sumin -\sqrt{n^{-1}\pmax\log(\pmax)}y_i \x_{\s,i}\tt\bu_\s \mathbf{1}\left(1-y_i\x_{\s,i}\tt\sbeta_\s\geq 0 \right)\right|\notag\\
&=\sqrt{n^{-1}\pmax\log(\pmax)}\sups\sup_{\|\bu_\s\|=\triangle}\left|\sumin  y_i\x_{\s,i}\tt\bu_\s\mathbf{1}\left(1-y_i\x_{\s,i}\tt\sbeta_\s\geq 0 \right) \right|\notag\\
&\leq\sqrt{n^{-1}\pmax\log(\pmax)} \triangle O_p\left\{\sqrt{\pmax+1}\sqrt{n\log(\pmax)}\right\}\notag\\
&= O_{p}(\triangle \pmax\log(\pmax)).\label{eq:341}
\end{align}
Turning to $B_{s,n2}$, under Conditions \ref{con:1viii} and \ref{con:1v} and according to \cite{Koo2008Abahadur}, $\Hess_\s(\bbeta_\s)$ is element-wise continuous at $\sbeta_s$. By Taylor expansion of the hinge loss at $\sbeta_\s$, we have
\begin{align}\label{eq:taylor}
\Hess_\s\left(\sbeta_\s+t\sqrt{n^{-1}\pmax}\bu_\s \right)=\Hess_\s(\sbeta_\s)+o(1).
\end{align}
Hence, it is shown that
\begin{align}
&\quad\infs\inf_{\|\bu_\s\|=\triangle}B_{s,n2}\notag\\
&=\infs\inf_{\|\bu_\s\|=\triangle}\sumin \left[\Exp\left\{\left(1-y_i\x_{\s,i}\tt(\sbeta_\s+\sqrt{n^{-1}\pmax\log(\pmax)}\bu_\s) \right)_+ \right\}-\Exp\left\{(1-y_i\x_{\s,i}\tt\sbeta_\s)_+\right\} \right]\notag\\
&= \infs\inf_{\|\bu_\s\|=\triangle} 2^{-1}\pmax\log(\pmax) \bu_\s\tt \Hess_\s\left(\sbeta_\s+t\sqrt{n^{-1}\pmax\log(\pmax)}\bu_\s \right)\bu_\s\notag\\
&\geq 2^{-1}\triangle^2 c_0\pmax\log(\pmax),  \label{eq:271}
\end{align}
for some $0<t<1$, where the last inequality is due to \eqref{eq:taylor} and Condition~\ref{con:1viii}.
It can be readily shown by \eqref{eq:211}, \eqref{eq:341}, \eqref{eq:271} and Condition \ref{con:1v} that when $\triangle$ is sufficiently large, $2^{-1}\triangle^2 c_0\pmax(>0)$ dominates other terms in $B_{s,n}$ . This completes the proof of Step 2.

Step 3:
Combining \eqref{eq:penalty}, \eqref{eq:222}, \eqref{eq:341} and \eqref{eq:271}, when $n$ and $\triangle$ are sufficiently large, we have
\begin{align}
&\quad\sups \inf_{\|\bu_\s\|=\triangle} \left\{ l_s\left(\sbeta_\s+\sqrt{n^{-1}p_{\max}\log(\pmax)}\bu_\s\right)-l_s(\sbeta_\s)\right\}\notag\\
&=\sups \inf_{\|\bu_\s\|=\triangle}\left\{ n^{-1}(A_{s,n}+B_{s,n})+2^{-1}\lambda_n \left\|\spbeta_\s +\sqrt{n^{-1}\pmax\log(\pmax)}\bu_\s^+\right\|^2-2^{-1}\lambda_n \|\spbeta_\s\|^2\right\}\notag\\
&\geq\sups \inf_{\|\bu_\s\|=\triangle}\Big\{ n^{-1}B_{s,n}-n^{-1}|A_{s,n}|-2^{-1}\lambda_n \Big|\left\|\spbeta_\s +\sqrt{n^{-1}\pmax\log(\pmax)}\bu_\s^+\right\|^2- \|\spbeta_\s\|^2\Big|\Big\}\notag\\
&\geq \infs \inf_{\|\bu_\s\|=\triangle}n^{-1} B_{s,n2}-\sups\sup_{\|\bu_\s\|=\triangle}n^{-1}|B_{s,n1}|-\sups\sup_{\|\bu_\s\|=\triangle}n^{-1}|A_{s,n}|\notag\\
&\quad-2^{-1}\lambda_n \sups\sup_{\|\bu_\s\|=\triangle} \Big|\left\|\spbeta_\s +\sqrt{n^{-1}\pmax\log(\pmax)}\bu_\s^+\right\|^2- \|\spbeta_\s\|^2\Big|\notag\\
&=2^{-1}n^{-1}\triangle^2 c_0\pmax\log(\pmax)- O_{p}(\triangle n^{-1} \pmax\log(\pmax))-\triangle^{3/2}n^{-1}\pmax o_{p}(1)\notag\\
&\quad -2^{-1}\triangle \lambda_n\pmax \sqrt{n^{-1}\log(\pmax)}\notag\\
&>0,
\end{align}
where the last inequality is obtained from
 { Conditions \ref{con:1viii}-\ref{con:1v} and $\lambda_n=O(\sqrt{n^{-1}\log(\pmax)})$}. This completes the proof of \eqref{eq:191}.

Equation \eqref{eq:9} can be proved in a similar way.  Note that $n-\lfloor n/J \rfloor\sim n$ and each sample from $\Dn$ is drawn independently from an identical distribution.  Hence $\tbeta^{[-j]}_\s$ converges to $\sbeta_\s$ in the same order as $\hbeta_\s$ for each $j=1,2,...,J$, i.e.,
\begin{align}
\max_{1\leq j\leq J} \sups \left\|\tbeta_\s^{[-j]}-\sbeta_\s \right\|=O_p\left(\sqrt{\frac{\pmax\log(\pmax)}{n}}\right).
 \end{align}

\end{proof}

\subsection{Proof of Theorem \ref{thm:optimality}} \label{sec:proof-of-opt}
Let us introduce Lemma \ref{lem:svm2} that facilitates the proof of Theorem \ref{thm:optimality}.
\begin{lem}\label{lem:svm2}
Assume that Condition \ref{con:xi} and
\begin{align}\label{eq:svm2-1}
\supw\left|\frac{\mathrm{CV}(\w)-R_n(\w)}{R_n(\w)}\right|=o_p(1)
\end{align}
hold. Then
\begin{align}
\frac{R_n(\hat{\w})}{\inf_{\w\in\mathcal{W}} R_n(\w)}\to 1\label{eq:lem2-opt}
\end{align}
in probability, where $\hat{\w}$ is the optimal solution from \eqref{eq:CV}.
\end{lem}

\begin{proof}
By the definition of infimum, there exist a sequence $\vartheta_n$ and a vector sequence $\w_n\in\calW$ such that as $n\to\infty$, $\vartheta_n\to 0$ and
\begin{align}
\infw R_n(\w)=R_n(\w_n)-\vartheta_n.
\end{align}
From Condition \ref{con:xi}, we have
\begin{align}
\frac{R_n(\w_n)}{\infw R_n(\w)}&>\frac{\infw R_n(\w)}{\infw R_n(\w)}=1,\label{eq:svm-143}\\
\intertext{and}
\frac{\vartheta_n}{\infw R_n(\w)}&=o_p(1).\label{eq:svm-144}
\end{align}
Taking \eqref{eq:svm2-1}, \eqref{eq:svm-143} and \eqref{eq:svm-144} together, for any $\delta>0$,
\begin{align}
&\quad\Pr\left\{ \left|\frac{\inf_{\w\in\mathcal{W}} R_n(\w)}{R_n(\hat{\w})}-1\right|>\delta\right\}\notag\\
&=\Pr\left\{\frac{R_n(\hat{\w})-\inf_{\w\in\mathcal{W}} R_n(\w)}{R_n(\hat{\w})}-1>\delta\right\}\notag\\
&=\Pr\left\{\frac{R_n(\hat{\w})-\CV(\hat{\w})+\CV(\hat{\w})-R_n(\w_n)+\vartheta_n}{R_n(\hat{\w})}>\delta\right\}\notag\\
&\leq\Pr\left\{\frac{R_n(\hat{\w})-\CV(\hat{\w})+\CV(\w_n)-R_n(\w_n)+\vartheta_n}{R_n(\hat{\w})}>\delta\right\}\notag\\
&\leq\Pr\left\{\frac{|R_n(\hat{\w})-\CV(\hat{\w})|}{R_n(\hat{\w})}
+\frac{|\CV(\w_n)-R_n(\w_n)|}{\infw R_n(\w)}+\frac{\vartheta_n}{\infw R_n(\w)}>\delta\right\}\notag\\
&\leq\Pr\left\{\supw\left|\frac{R_n(\w)-\CV(\w)}{R_n(\w)}\right|
+\frac{|\CV(\w_n)-R_n(\w_n)|/R_n(\w_n)}{\infw R_n(\w)/R_n(\w_n)}+\frac{\vartheta_n}{\infw R_n(\w)}>\delta\right\}\notag\\
&\leq\Pr\Bigg\{\supw\left|\frac{R_n(\w)-\CV(\w)}{R_n(\w)}\right|
+\supw\left|\frac{\CV(\w)-R_n(\w)}{R_n(\w)}\right|\frac{R_n(\w_n)}{\infw R_n(\w)}+\frac{\vartheta_n}{\infw R_n(\w)}>\delta\Bigg\}\notag\\
&\to 0,\label{eq:48}
\end{align}
which implies that \eqref{eq:lem2-opt} is valid.
\end{proof}
\begin{proof}[Theorem \ref{thm:optimality}]
 Let
\begin{align}
T_n=\onen\sumin \left(1-y_i\x_i\tt\hbeta(\w)\right)_+
\end{align}
By Lemma \ref{lem:svm2} and the triangle inequality, it suffices to verify that
\begin{align}
\supw\frac{|\CV(\w)-T_n(\w)|}{R_n(\w)}=o_p(1),\label{eq:main1}
\end{align}
and \\
\begin{align}
\supw\frac{|T_n(\w)-R_n(\w)|}{R_n(\w)}=o_p(1).\label{eq:main2}
\end{align}

For \eqref{eq:main1}, we have
\begin{align}
|\CV(\w)-T_n(\w)|&=\left|\onen\sumj\sum_{i\in\calA(j)}\left\{(1-y_i\x_i\tt \wbeta^{[-j]}(\w))_+-(1-y_i\x_i\tt \hbeta(\w))_+\right\}\right|\notag\\
&\leq \onen\sumj\sum_{i\in\calA(j)}\left|
\int_{y_i\x_i\tt \wbeta^{[-j]}(\w)}^{y_i\x_i\tt \hbeta(\w)}I(t\leq 1)\diff t
\right|\notag\\
&\leq \onen\sumj\sum_{i\in\calA(j)}\left| y_i\x_i\tt \left(\wbeta^{[-j]}(\w)-\hbeta(\w)\right)\right|\notag\\
&\leq\onen\sumj\sum_{i\in\calA(j)}\sums w_s\|\x_{\s,i}\|\left\|\wbeta_\s^{[-j]}-\hbeta_\s\right\|\notag\\
&\leq \frac{1}{n}\sumin\sups\|\x_{\s,i}\|\max_{1\leq j\leq J}\sups\left\|\wbeta_\s^{[-j]}-\hbeta_\s\right\|\notag\\
&\leq C_1\sqrt{\pmax}\max_{1\leq j\leq J}\sups\left\|\wbeta_\s^{[-j]}-\hbeta_\s\right\|\notag\\
&=O_p\left(\frac{\pmax\sqrt{\log(\pmax)}}{\sqrt{n}}\right)\notag\\
&=o_p(1),\label{eq:|cv-ln|}
\end{align}
where the second last equality is established based on Lemma~\ref{lem:consist}, and  the last equality is based on Conditions \ref{con:1v}. Coupled with Condition~\ref{con:xi} and \eqref{eq:|cv-ln|}, we obtain \eqref{eq:main1}.

To prove \eqref{eq:main2}, note that
\begin{align}
| T_n(\w)-R_n(\w)|
&=\left|\onen\sumin \left(1-y_i\x_i\tt\hbeta(\w) \right)_+-\Exp\left\{(1-y\x\tt\hbeta(\w) )_+|\Dn\right\}\right|\notag\\
&\leq\left|\onen\sumin \left(1-y_i\x_i\tt\hbeta(\w) \right)_+-\onen\sumin \left(1-y_i\x_i\tt\sbeta(\w) \right)_+\right|\notag\\
&\quad+\left|\onen\sumin \left(1-y_i\x_i\tt\sbeta(\w) \right)_+-\Exp\left(1-y\x\tt\sbeta(\w)\right)_+\right|\notag\\
&\quad+\left|\Exp\left(1-y\x\tt\sbeta(\w)\right)_+-\Exp\left\{(1-y\x\tt\hbeta(\w))_+ |\Dn\right\}\right|\notag\\
&\equiv |\Omega_1(\w)|+|\Omega_2(\w)|+|\Omega_3(\w)|.\label{eq:20}
\end{align}
Recognising the above, Lemma~\ref{lem:consist} and Conditions \ref{con:1ii2} and \ref{con:1v}, it can be shown that
\begin{align}
 \supw|\Omega_1(\w)|
 &\leq \supw\onen\sumin \left| \left(1-y_i\x_i\tt\hbeta(\w) \right)_+- \left(1-y_i\x_i\tt\sbeta(\w) \right)_+\right|\notag\\
 &\leq \supw\onen\sumin \left| y_i\x_i\tt\left(\sbeta(\w)-\hbeta(\w)\right)\right|\notag\\
 &\leq \supw\onen\sumin\sums w_s \|\x_{\s,i}\|\left\| \sbeta_\s-\hbeta_\s\right\|\notag\\
 &\leq  \sups\left\| \sbeta_\s-\hbeta_\s\right\|\supi\sups\|\x_{\s,i}\|\notag\\
 &=O_p\left(\frac{\pmax\sqrt{\log(\pmax)}}{\sqrt{n}}\right)\notag\\
&=o_p(1).\label{eq:Omega1}
\end{align}
{
 Define
\begin{align}
|\w-\w'|_1=\sums|w_s-w'_s|,
\end{align}
for any $\w=(w_1,...,w_{S_n})\in\calW$ and $\w'=(w'_1,...,w'_{S_n})\in\calW$.}
Let $h_n=1/(
\pmax\log n)$ and create grids using regions of the form $\calW^{(l)}=\{\w:|\w-\w^{(l)}|_1\leq h_n\}$.
By the notion of the $\epsilon-$covering number introduced by \cite{vaart1996weakconvergence}, $\calW$ can be covered with $N=O(1/h_n^{S_n-1})$ regions $\calW^{(l)}$, $l=1,...,N.$

Note that
\begin{align}
&\quad\sup_{\w\in\calW^{(l)}}|\Omega_2(\w)-\Omega_2(\w^{(l)})|\notag\\
&\leq\sup_{\w\in\calW^{(l)}}\left|\onen\sumin \left(1-y_i\x_i\tt\sbeta(\w) \right)_+-\onen\sumin \left(1-y_i\x_i\tt\sbeta(\w^{(l)}) \right)_+\right|\notag\\
&\quad+\sup_{\w\in\calW^{(l)}}\left|\Exp\left(1-y\x\tt\sbeta(\w)\right)_+-\Exp\left(1-y\x\tt\sbeta(\w^{(l)})\right)_+\right|\notag\\
&\leq\sup_{\w\in\calW^{(l)}}\onen\sumin\left|y_i\x_i\tt\{\sbeta(\w^{(l)})-\sbeta(\w)\} \right|+\sup_{\w\in\calW^{(l)}}\Exp\left|y\x\tt\{\sbeta(\w^{(l)})-\sbeta(\w)\} \right|\notag\\
&\leq\sup_{\w\in\calW^{(l)}}\onen\sumin \sums|w_s-w^{(l)}_s|\left|\x_{\s,i}\tt\sbeta_\s\right|+\sup_{\w\in\calW^{(l)}}\sums|w_s-w^{(l)}_s|\Exp\left|\x_{\s,i}\tt\sbeta_\s\right|\notag\\
&=\sup_{\w\in\calW^{(l)}}|\w-\w^{(l)}|_1\sups\|\sbeta_\s\|\left(\supi\sups \|\x_{\s,i}\|+
\supi\sups\Exp\|\x_{\s,i}\| \right)\notag\\
&\leq  \frac{C_2\sqrt{\pmax}}{\pmax\log(n)}2C_1\sqrt{\pmax}\notag\\
&=O_p\left( \log^{-1}(n)\right)\notag\\
&=o_p(1),
\end{align}
where the result holds uniformly for $j$.
Hence we have
\begin{align}
\supw|\Omega_2(\w)|&=\supjN\supjw|\Omega_2(\w)|\notag\\
&\leq \supjN|\Omega_2(\w^{(l)})|+\supjN\supjw|\Omega_2(\w)-\Omega_2(\w^{(l)})|\notag\\
&=\supjN|\Omega_2(\wj)|+o_p(1).\label{eq:26}
\end{align}
Furthermore, for any $\epsilon>0$,
\begin{align}
&\quad\Pr\left\{\supjN |\Omega_2(\wj)|> 3\epsilon\right\}\notag\\
&=\Pr\Bigg[\supjN\Big|\onen\sumin\left(1-y_i\x_i\tt\sbeta(\wj)\right)_+\mathbf{1}\left(|1-y_i\x_i\tt\sbeta(\wj)|<\pmax n^{0.1}\right)\notag\\
&\quad+\onen\sumin\left(1-y_i\x_i\tt\sbeta(\wj)\right)_+\mathbf{1}\left(|1-y_i\x_i\tt\sbeta(\wj)|\geq \pmax n^{0.1}\right)\notag\\
&\quad-\Exp\left\{(1-y\x\tt\sbeta(\wj))_+\mathbf{1}\left(|1-y\x\tt\sbeta(\wj)|< \pmax n^{0.1}\right)\right\}\notag\\
&\quad-\Exp\left\{(1-y\x\tt\sbeta(\wj))_+\mathbf{1}\left(|1-y\x\tt\sbeta(\wj)|\geq \pmax n^{0.1}\right)\right\}\Big|> 3\epsilon\Bigg]\notag\\
&\leq \Pr\Bigg[\supjN\Big|\onen\sumin\left(1-y_i\x_i\tt\sbeta(\wj)\right)_+\mathbf{1}\left(|1-y_i\x_i\tt\sbeta(\wj)|<\pmax n^{0.1}\right)\notag\\
&\quad -\Exp\left\{(1-y\x\tt\sbeta(\wj))_+\mathbf{1}\left(|1-y\x\tt\sbeta(\wj)|< \pmax n^{0.1}\right)\right\}\Big|>\epsilon\Bigg]\notag\\
&\quad+\Pr\Bigg[\supjN\onen\sumin\left(1-y_i\x_i\tt\sbeta(\wj)\right)_+\mathbf{1}\left(|1-y_i\x_i\tt\sbeta(\wj)|\geq \pmax n^{0.1}\right)>\epsilon\Bigg]\notag\\
&\quad+\Pr\Bigg[\supjN\Exp\left\{(1-y\x\tt\sbeta(\wj))_+\mathbf{1}\left(|1-y\x\tt\sbeta(\wj)|\geq \pmax n^{0.1}\right)\right\}> \epsilon\Bigg]\notag\\
&\equiv \Xi_1+\Xi_2+\Xi_3.\label{eq:27}
\end{align}
Clearly,
\begin{align}
&\quad\sumin\Exp\left\{(1-y_i\x_i\tt\sbeta(\wj))_+\right\}^2\notag\\
&\leq \sumin\Exp\left|1-y_i\x_i\tt\sbeta(\wj)\right|^2\notag\\
&\leq \sumin\Exp\left(1+2|\x_i\tt\sbeta(\w^{(l)})|+\x_i\tt\sbeta(\w^{(l)}) \bbeta^{*\text{T}}(\w^{(l)})\x_i\tt\right)\notag\\
&\leq \sumin\Exp\left(1+2\sups \|\x_{\s,i}\|\|\sbeta_\s\|+\sups\|\sbeta_\s\|^2\|\x_{\s,i}\|^2 \right)\notag\\
&\leq 4C^2_1C^2_2 n p^2_{\max}.\label{eq:471}
\end{align}
Using Boole's and Bernstein's inequalities and by \eqref{eq:471},
\begin{align}
\Xi_1&\leq\sum_{j=1}^N\Pr\Bigg[\Big|\onen\sumin\left(1-y_i\x_i\tt\sbeta(\wj)\right)_+\mathbf{1}\left(|1-y_i\x_i\tt\sbeta(\wj)|<\pmax n^{0.1}\right)\notag\\
&\quad -\Exp\left\{(1-y\x\tt\sbeta(\wj))_+\mathbf{1}\left(|1-y\x\tt\sbeta(\wj)|< \pmax n^{0.1}\right)\right\}\Big|>\epsilon\Bigg]\notag\\
&\leq N\exp\left(- \frac{n^2\epsilon^2/2}{4C_1^2C_2^2 np_{\max}^2+\epsilon \pmax n^{0.1}/3}\right)\notag\\
&\leq (\pmax\log n)^{S_n-1}\exp\left(- \frac{n^2\epsilon^2/2}{4C_1^2C_2^2 np_{\max}^2+\epsilon \pmax n^{0.1}/3}\right)\notag\\
&=O\left\{\exp\left(-\epsilon^2 n p^{-2}_{\max}+ S_n \log(\pmax)+S_n\log \log(n)\right)\right\}\notag\\
&=o(1),\label{eq:29}
\end{align}
where the last equality is established from Condition~\ref{con:1v} and {the condition that $S_n=O(n^\tau)$ for $\tau\in (0,1-2\kappa)$.}
Additionally, we can write
\begin{align}
\Xi_2&=\Pr\Bigg\{\supjN\onen\sumin\left(1-y_i\x_i\tt\sbeta(\wj)\right)_+\mathbf{1}\left(|1-y_i\x_i\tt\sbeta(\wj)|\geq \pmax n^{0.1}\right)>\epsilon\Bigg\}\notag\\
&\leq \Pr\left(\supjN\supi|1-y_i\x_i\tt\sbeta(\wj)|\geq \pmax n^{0.1}\right)\notag\\
&\leq\Pr\left\{\supjN\supi\sums{w^{(l)}_s}\left(1+\|\x_{\s,i}\|\|\sbeta_\s\|\right)\geq \pmax n^{0.1}\right\}\notag\\
&\leq \Pr\left\{\left(1+\supi\sups\|\x_{\s,i}\|\sups\|\sbeta_\s\|\right)\geq \pmax n^{0.1}\right\}\notag\\
&= o(1),\label{eq:30}
\end{align}
where the last inequality holds because of Conditions \ref{con:1ii} and \ref{con:1ii2}. Similarly,
\begin{align}
\Xi_3&=\Pr\left[\supjN\Exp\left\{(1-y\x\tt\sbeta(\wj))_+\mathbf{1}\left(|1-y\x\tt\sbeta(\wj)|\geq \pmax n^{0.1}\right)\right\}> \epsilon\right]\notag\\
&\leq \Pr\left(\supjN\Exp|1-y\x\tt\sbeta(\wj)|\geq \pmax n^{0.1}\right)\notag\\
&\leq \Pr\left\{\left(1+\sups\Exp\|\x_\s\| \sups\|\sbeta_\s\|\right)\geq\pmax n^{0.1}\right\}\notag\\
&=o(1).\label{eq:31}
\end{align}
Together with \eqref{eq:27}, \eqref{eq:29}-- \eqref{eq:31}, we obtain $\supjN|\Omega_2(\wj)|=o_p(1)$.  As well, by \eqref{eq:26}, we have
\begin{align}
\supw|\Omega_2(\w)|=o_P(1).\label{eq:312}
\end{align}

Finally, note that $(y,\x)$ and $(\tilde{y}, \tilde{\x})$ are independently and identically distributed, and under Lemma~\ref{lem:consist}, we have
\begin{align}
 \supw|\Omega_3(\w)|
 &=\supw\left|\Exp\left(1-y\x\tt\sbeta(\w)\right)_+-\Exp\left\{(1-\tilde{y}\tilde{\x}\tt\hbeta(\w)) |\Dn\right\}_+\right|\notag\\
 &=\supw\left|\Exp\left(1-\tilde{y}\tilde{\x}\tt\sbeta(\w)\right)_+-\Exp\left\{(1-\tilde{y}\tilde{\x}\tt\hbeta(\w)) |\Dn\right\}_+\right|\notag\\
 &= \supw\left|
 \Exp\int_{\tilde{y}\tilde{\x}\tt\sbeta(\w)}^{\tilde{y}\tilde{\x}\tt\hbeta(\w)}I(t\leq 1)\diff t\bigg|\Dn\right|\notag\\
 &\leq \supw\Exp\left\{\left|\tilde{y}\tilde{\x}\tt\left(\hbeta(\w)-\sbeta(\w)\right) \right| \big|\Dn\right\}\notag\\
 &\leq \supw \sums w_s\Exp \left\{\left|\tilde{\x}_\s\tt\left(\hbeta_\s-\sbeta_\s\right)\right|\big|\Dn\right\}\notag\\
 &\leq \sups \left\|\hbeta_\s-\sbeta_\s \right\|\sups\Exp\|\tilde{\x}_{\s,i}\|\notag\\
 &=O_p\left(\frac{ \pmax\sqrt{\log(\pmax)}}{\sqrt{n}}\right)\notag\\
  &=o_p(1),\label{eq:II3}
\end{align}
where the last inequality holds due to Condition \ref{con:1v}. Putting \eqref{eq:20}, \eqref{eq:Omega1}, \eqref{eq:312} and \eqref{eq:II3} together, we complete the proof of \eqref{eq:main2}.

\end{proof}

\vskip 0.2in
\bibliography{mybib,MA_SIM_ref2}

\end{document}